\documentclass[11pt,a4paper,oneside]{article}
\newcommand{\ignore}[1]{}
\normalbaselines 
\usepackage{amsmath}
\usepackage{amssymb}
\usepackage{mathtools}
\usepackage{mathcomp}
\usepackage{textcomp}
\usepackage{amsthm}
\usepackage{enumerate}
\usepackage[auth-lg]{authblk}
\usepackage{cite}
\usepackage{tikz}
\usepackage{caption}
\usepackage{subcaption}
\usepackage[ruled,linesnumbered]{algorithm2e}
\usetikzlibrary{arrows}
\usepackage{multicol}
\usepackage{tikz-qtree}
\usepackage{rotating}
\usepackage{url}
\usepackage{array}

\ignore{
\bibliographystyle{abbrv}
\usepackage[dvipdf]{graphicx}
\usepackage{bm}
\usepackage{amsmath,amsthm,amssymb,url}
\usepackage{multirow,bigdelim}
\usepackage{amscd}
}
\newtheorem{thm}{Theorem}[section]\theoremstyle{plain}
\newtheorem{theorem}[thm]{Theorem}\theoremstyle{plain}
\theoremstyle{plain}
\newtheorem{lemma}[thm]{Lemma}\theoremstyle{plain}
\newtheorem{definition}[thm]{Definition}\theoremstyle{plain}

\theoremstyle{plain}
\newtheorem{claim}[thm]{Claim}\theoremstyle{plain}
\newtheorem{corollary}[thm]{Corollary}\theoremstyle{plain}
\theoremstyle{plain}
\theoremstyle{plain}
\theoremstyle{plain}
\theoremstyle{plain}
\theoremstyle{plain}
\theoremstyle{plain}

\newcommand{\MS}{{\mathcal S}}
\newcommand{\MI}{{\mathcal I}}
\newcommand{\MF}{{\mathcal F}}

\newcommand{\MC}{{\mathcal C}}

\newcommand{\MH}{{\mathcal H}}
\newcommand{\MX}{{\mathcal X}}
\newcommand{\MY}{{\mathcal Y}}

\usepackage{xcolor}
\usepackage{graphicx} % Allows including images
\usepackage{booktabs} % Allows the use of \toprule, \midrule and \bottomrule in tables
\usepackage[ruled,linesnumbered]{algorithm2e}
\usepackage{pgf,tikz}
\usetikzlibrary{positioning}
\usepackage{comment}

\ignore{
\usetikzlibrary{arrows,matrix,positioning, fit, backgrounds, calc,shapes.geometric}
\tikzstyle{mybox} = [draw=red, very thick,
    rectangle, rounded corners, inner sep=10pt, inner ysep=20pt]
\tikzstyle{fancytitle} =[fill=red, text=white]
\tikzset{
    invisible/.style={opacity=0,text opacity=0},
    visible on/.style={alt={#1{}{invisible}}},
    alt/.code args={<#1>#2#3}{
      \alt<#1>{\pgfkeysalso{#2}}{\pgfkeysalso{#3}} 
    },
  }
\definecolor{ForestGreen}{rgb}{0.13, 0.55, 0.13}
\definecolor{FluorescentOrange}{rgb}{1.0, 0.75, 0.0}
\tikzset{%
  highlight/.style={rectangle,rounded corners,fill=green!15,draw,fill opacity=0.5,thick,inner sep=0pt}
}

\usetikzlibrary{patterns}
}

\begin{document}

%\title{\bf A certificate for Robinsonian matrices}
\title{A Structural Characterization  for Certifying Robinsonian Matrices}

 \author[1,2]{Monique Laurent}
\author[1]{Matteo Seminaroti}
\author[1,3]{Shin-ichi Tanigawa}
\date{}

\affil[1]{\small Centrum Wiskunde \& Informatica (CWI), Science Park 123, 1098 XG Amsterdam, The Netherlands}
\affil[2]{\small Tilburg University, P.O. Box 90153, 5000 LE Tilburg, The Netherlands}
\affil[3]{\small Research Institute for Mathematical Sciences, Kyoto University, Kitashirakawa-Oiwaketyo, Sakyo-ku, Kyoto, 606-8502, Japan}
\renewcommand{\thefootnote}{\fnsymbol{footnote}}
\footnotetext[1]{Correspondence to : \texttt{M.Laurent@cwi.nl} (M.~Laurent, CWI, Postbus 94079, 1090 GB, Amsterdam. Tel.:+31 (0)20 592 4386), \texttt{tanigawa@kurims.kyoto-u.ac.jp} (S. Tanigawa, CWI, Postbus 94079, 1090 GB, Amsterdam. Tel.:+31 (0)20 592 9333).}

\graphicspath{{./figures/}}

\maketitle

\begin{abstract}
A symmetric matrix is Robinsonian  if its rows and columns can be simultaneously reordered in such a way that entries are monotone nondecreasing in rows and columns when moving toward the diagonal. The adjacency matrix of a graph  is Robinsonian precisely when the graph is a unit interval graph, so that  Robinsonian matrices form a matrix analogue of the class of unit interval graphs. 
Here we provide a structural characterization for Robinsonian matrices in terms of forbidden substructures, extending the notion of  asteroidal triples to weighted graphs. This implies  the known characterization of unit interval graphs and leads to an efficient algorithm for certifying that a matrix is not Robinsonian.

%\tcolred{which  thus complements the existing efficient approaches for certifying Robinsonian matrices by returning a Robinson ordering.}

\medskip

\end{abstract}

%%%%%%%%%%%%%%%%%%%%%%%%%%%%%%%%%%%%%%%%%%%%%%%%%
\section{Introduction} 
%%%%%%%%%%%%%%%%%%%%%%%%%%%%%%%%%%%%%%%%%%%%%%%%%

\subsection{Background}
Robinsonian matrices are a special class of structured matrices introduced by Robinson~\cite{Robinson51} to model the seriation problem, a combinatorial problem arising in data analysis, which asks to  sequence a set of objects according to their pairwise similarities in such a way that similar objects are ordered close to each other.
Seriation has many applications in a wide range of different subjects, including archaeology~\cite{Petrie99,Kendall69}, data visualization and exploratory analysis \cite{Hubert01,Brusco05}, bioinformatics (e.g., microarray gene expression \cite{Tien08}), machine learning (e.g., to pre-estimate the number and shape of clusters  \cite{Ding04,Havens12}). We refer to the survey  \cite{Innar10} for details  and further references.
%It was originally introduced for applications arising in archaeology~\cite{Petrie99,Kendall69}, but it can be used, among others, for data visualization and exploratory analysis \cite{Hubert01,Brusco05}, with applications in bioinformatics (e.g., microarray gene expression \cite{Tien08}). In machine learning, seriation is used to pre-estimate the number of clusters or the tendency of data patterns to form clusters \cite{Havens12}, with applications to text mining~\cite{Ding04}.
%For a more exhaustive and complete list of applications of seriation we refer the interested reader to the survey \cite{Innar10}.

A symmetric $n \times n$ matrix~$A$ is called  a {\em Robinson similarity} if its entries are monotone nondecreasing 
 in the rows and columns when moving {toward} the main diagonal, i.e., if \index{Robinson! similarity} \index{Robinsonian! similarity}
\begin{equation}\label{eq:1-Robinson}
A_{xz}\le \min\{A_{xy},A_{yz}\} % \quad \text{for each} \quad 1 \le  x < y < z \le n.
\end{equation}
for all $1\le x<y<z\le n$. 
Throughout we will call any ordered triple $(x,y,z)$ satisfying the condition (\ref{eq:1-Robinson}) a {\em Robinson triple}.
If the rows and columns of $A$ can be symmetrically reordered by a permutation~$\pi$ in such a way that the permuted matrix is  a Robinson similarity, then $A$ is said to be a {\em Robinsonian similarity} and~$\pi$ is called a {\em Robinson ordering} of $A$. 
By construction, $\pi$ is a Robinson ordering of $A$ if any triple $(x,y,z)$ ordered in $\pi$ as  $x \prec_{\pi} y \prec_{\pi} z$ is Robinson.
%, so  the objects $x,y$ (which are ordered closer in $\pi$ than objects $x,z$) are more similar to each other than objects $x,z$ and the same holds for the objects $y,z$.
Hence Robinsonian matrices best achieve the goal of seriation, which is  to order similar objects close to each other.

%\medskip
%Some  natural questions arise: How to recognize whether a given matrix $A$ is Robinsonian? if $A$ is Robinsonian, how to find a Robinson ordering of $A$?  
%If $A$ is not Robinsonian,  what is a simple  certificate for this?

\medskip
Several Robinsonian recognition algorithms are known, permitting to check whether a matrix is Robinsonian in polynomial time. 
Most of the known algorithms rely on characterizations of Robinsonian matrices in terms of interval (hyper)graphs and the consecutive-ones property (C1P).
% and they use the data structure of PQ-trees of Booth and Leuker~\cite{Booth76} to return a Robinson ordering. 
Specifically, call a {\em ball} of $A$ any set of the form $B(x,\rho):=\{y\in V: A_{xy}\ge \rho\}$ for some $x\in V$ and scalar $\rho>0$.
Let the \textit{ball hypergraph} of $A$ be the hypergraph whose vertex set is $V$ and with hyperedges the balls of $A$. Then $A$ is Robinsonian if and only if its ball hypergraph is an interval hypergraph~\cite{Mirkin84}.
Equivalently, define the {\em ball intersection graph} of $A$ as the graph whose vertex set is the set  of balls of $A$, with an edge between two distinct balls if they intersect.
Then, it is known that  $A$ is a Robinsonian matrix if and only if its ball intersection graph is an interval graph (see \cite{Prea14}). 
Using the above characterizations, distinct recognition algorithms were introduced by Mirkin and Rodin \cite{Mirkin84} with running time  $O(n^4)$, by Chepoi and Fichet \cite{Chepoi97} with  running  time  $O(n^3)$, and by Pr\'ea and Fortin~\cite{Prea14} with running time $O(n^2)$, when applied to a $n \times n$ symmetric matrix.

Different algorithms were recently introduced in  \cite{Laurent15,Laurent16}, based on a link between Robinsonian matrices and unit interval graphs (pointed out in \cite{Roberts69}) and exploiting
the fact that unit interval graphs can be recognized efficiently using a simple graph search algorithm, namely Lexicographic Breadth-First Search (Lex-BFS) (see \cite{Corneil95,Corneil04}).
The algorithm of \cite{Laurent15} is based on expressing Robinsonian matrices as conic combinations of (adjacency matrices of) unit interval graphs and iteratively using Lex-BFS to check whether these are unit interval graphs; its overall running time is $O(L(m+n))$, where $L$ is the number of distinct values in the matrix and $m$ is its number of nonzero entries.
The algorithm of \cite{Laurent16} relies 
on a new search algorithm, Similarity-First Search (SFS), which can be seen as a generalization of Lexicographic Breadth-First Search (Lex-BFS) to the setting of weighted graphs. The SFS algorithm 
runs in $O(n+m\log n)$ time and the recognition algorithm for Robinsonian matrices terminates after at most $n$ iterations of SFS,  thus with  overall running time  $O(n^2 + nm\log n)$ \cite{Laurent16}. 

%When applied to a Robinsonian matrix $A$,  the recognition algorithm returns a Robinson ordering of $A$, and otherwise the algorithm just returns that $A$ is not Robinsonian. 

\medskip
All the current recognition algorithms return a certificate (i.e., a Robinson ordering) only if the matrix is Robinsonian, and otherwise they just return a negative answer.
In this paper we give a new structural characterization of Robinsonian matrices in terms of forbidden substructures. 
We provide a simple, succinct certificate for {\em non}-Robinsonian matrices,
which represents a natural extension of the known structural characterization for unit interval graphs. 
Specifically, our certificate involves the new notion of {\em weighted} asteroidal triple,  which generalizes to the matrix setting the known obstructions for unit interval graphs (namely, chordless cycles, claws and asteroidal triples), see Section \ref{secmainthm} for details.
Moreover we also give a simple, efficient algorithm for finding such a certificate, running in time $O(n^3)$ for a matrix of size $n$.

Observe that other certificates could be obtained from the alternative characterizations of Robinsonian matrices in terms of interval (hyper)graphs. Indeed, as the minimal obstructions for interval graphs are known (namely, they are the chordless cycles and the 
%graphical
asteroidal triples), we can derive from this an alternative  structural characterization for Robinsonian matrices. However this characterization is  expressed in terms of the ball intersection graph of $A$, whose vertex set is the set of balls rather than the index set of $A$, and thus it is not directly in terms of $A$ as in our main result (Theorem \ref{thm:main} below).

In the rest of the introduction, we first recall some properties of unit interval graphs, which will also serve as motivation for the notions and results we will introduce for Robinsonian matrices, and then we state our main structural result for Robinsonian matrices.

\subsection{Structural characterization of unit interval graphs}

Recall that a graph $G=(V=[n],E)$ is a unit interval graph if one can label its vertices by unit intervals in such a way that adjacent vertices correspond to intersecting intervals.
%The following link was established  between unit interval graphs and Robinsonian matrices by Roberts \cite{Roberts69}, who observed that
%a graph $G$ is a unit interval graph if and only if its adjacency matrix $A_G$ is a Robinsonian similarity matrix.
Roberts \cite{Roberts69} observed that a graph $G$ is a unit interval graph if and only if its adjacency matrix $A_G$ is a Robinsonian similarity matrix.

% if $A_G$ denotes the adjacency matrix of a graph $G$,}
%then $A_G$ is a Robinsonian similarity if and only if $G$ is a unit interval graph. % \cite{Roberts69}.
In particular, $G$ is a unit interval graph if and only if  there exists an ordering $\pi$ of its vertices such that $\{x,z\}\in E$ implies $\{x,y\}\in E$ and $\{y,z\}\in E$ whenever $x\prec_\pi y\prec_\pi z$. This condition,  known as the {\em 3-vertex condition},  is thus a specialization of the Robinson condition (\ref{eq:1-Robinson}) (see, e.g., \cite{Looges93}).

%Testing whether a given graph $G$ is a unit interval graph can be done in polynomial time. The recognition algorithm of Corneil \cite{Corneil04} is particularly simple, as it relies only on  a basic graph search, namely on Lexicographic Breadth-First Search  (Lex-BFS) and its variant Lex-BFS$_+$. 
%More precisely, Corneil's recognition algorithm  consists of one application of Lex-BFS followed by two iterations  of Lex-BFS$_+$, with overall running time $O(|V|+|E|)$. 
%This algorithm was the motivation for the new Robinsonian recognition algorithm in \cite{Laurent16}, based on the extension of Lex-BFS by the SFS algorithm.

\medskip
We now mention an alternative characterization  of  unit interval graphs   in terms of forbidden substructures. We recall some definitions.
Let $G=(V,E)$ be a graph.
Given $x,y,z\in V$, a {\em path from $x$ to $y$ missing\footnote{Sometimes one also says that the path  {\em avoids} $z$ (e.g. in \cite{Corneil98}). We use here the word ``miss" instead of ``avoid", in order to keep the word ``avoid" for the context of matrices and to prevent possible confusion.}  $z$} is a path $P=(x=x_0,x_1,\cdots, x_k, y=x_{k+1})$ in $G$ which is disjoint from the neighborhood of $z$, i.e., all pairs $\{x_i,x_{i+1}\}$ ($0\le i\le k$) are edges of $G$ and $z$ is not adjacent to any node of $P$.
%In other words, with $A_G$  the adjacency matrix of $G$,  $(A_G)_{x_i,x_{i+1}}>\max\{(A_G)_{x_iz},(A_G)_{x_{i+1}z}\}$ for all $0\le i\le k$.
%Hence if $P$ $G$-avoids $z$ in $G$ it also avoids $z$ in $A_G$, but the converse is not true in general.
An {\em asteroidal triple} in $G$ is  a set of nodes $\{x,y,z\}$ 
which is independent in $G$ (i.e., induces no edge) and 
such that between any two nodes in $\{x,y,z\}$ there exists a path in $G$ between them which misses the third one.
% In this paper we will  also call such a triple a {\em graphical asteroidal triple} in order to distinguish with the notion of {\em weighted asteroidal triple} which we will introduce later for matrices. 
%Hence, if $\{x,y,z\}$ is a graphic asteroidal triple in $G$,  then it is also a weighted asteroidal triple in the adjacency matrix $A_G$ of $G$, but the converse is not true in general.
Asteroidal triples  were introduced  for the recognition of interval graphs in \cite{Lekkerkerker62}, where it is shown that a graph is an interval graph if and only if it is chordal and does not have an asteroidal triple.  As unit interval graphs are precisely the claw-free interval graphs (see \cite{Roberts69}) we get the following characterization.

\begin{theorem}\label{thm:uig}(see \cite{Roberts69,Gardi07})
A graph $G$ is a unit interval graph if and only if it satisfies the following three  conditions:
\begin{description}
\item[(i)] $G$ is chordal, i.e., $G$ does not contain an induced cycle of length at least 4;
\item[(ii)] $G$ does not contain an  induced claw $K_{1,3}$;
\item[(iii)]
$G$ does not contain an asteroidal triple.
\end{description}
\end{theorem}

\subsection{Structural characterization of Robinsonian matrices}\label{secmainthm}

We now extend the above notion of asteroidal triple to the general setting of matrices and we then use it to state our main structural characterization  for Robinsonian matrices.

Let $V$ be a finite set and let  $A$ be a symmetric matrix indexed by $V$. 
Given $z\in V$, 
 a  {\em path  avoiding $z$ in $A$ %consists of a sequence of contiguous edges avoiding $z$; in other words 
 is of the form $P=(v_0,v_1,\cdots, v_k)$, where $v_0,\cdots,v_k$ are distinct elements of $V$ and, for each $1\le i\le k$,  the triple
$(v_{i-1},z,v_{i})$ is not Robinson, i.e.,  $A_{v_{i-1}v_i}>\min\{A_{v_{i-1}z},A_{v_iz}\}$. }
%We also say that the path  $P=(v_0,v_1,\cdots, v_k)$  goes from $v_0$ to $v_k$ and we set $V(P)=\{v_0,\cdots, v_k\}$.
Throughout, for distinct elements $x,y,z\in V$, we will use the notation $x\overset{z}{\sim}y$ to denote that  there exists  a path from  $x$ to $y$ avoiding $z$ in $A$. 

This concept was introduced in \cite[Definition 2.3]{Laurent16}, where it is  used as a key tool for analyzing the new recognition  algorithm for  Robinsonian matrices. Indeed,  
%\tcolred{as was observed  in ~\cite[Lemma 2.4]{Laurent16},}
saying that  the pair $(x,y)$ avoids $z$ means that   the triple $(x,z,y)$ is not Robinson (and the same for its reverse $(y,z,x)$), so that $z$ cannot be placed between $x$ and $y$ in {\em any} Robinson ordering of $A$.
An important consequence, as observed in~\cite[Lemma 2.4]{Laurent16}, is then that if there exists a path from $x$ to $y$ avoiding $z$, i.e., $x\overset{z}{\sim} y$, then $z$ cannot be placed between $x$ and $y$ in any Robinson ordering of $A$.

Therefore, if there exist three distinct elements $x, y, z$ with $x\overset{z}{\sim}y,\   y\overset{x}{\sim}z$ and $z\overset{y}{\sim}x$ in $A$, then we can conclude that no Robinson ordering of $A$ exists and thus that $A$ is not a Robinsonian similarity matrix.
This motivates the following definition.

\begin{definition}\label{defAT}
Let $A$ be a symmetric matrix. 
A triple $\{x,y,z\}$ of distinct elements of $V$ is called a {\em weighted asteroidal triple of $A$} %, also denoted as a $\WAT$ for short,
   if $x\overset{z}{\sim}y,  \ y\overset{x}{\sim}z$ and $ z\overset{y}{\sim}x$ hold, i.e., for any two elements of $\{x,y,z\}$ there exists a path between them avoiding the third one.
   % That is, there exist in $A$ three paths $P,Q,R$, where $P$ is a path from $x$ to $y$ avoiding $z$, $Q$ is a path from $x$ to $z$ avoiding $y$, and $R$ is a path from $y$ to $z$ avoiding $x$. 
%We say that $A$ is $\WAT$-free if no weighted asteroidal triple of $A$ exists.
\end{definition}

Our main result is that weighted asteroidal triples are the {\em only} obstructions to the Robinsonian property.

\begin{theorem}
\label{thm:main}
A symmetric matrix $A$ is a Robinsonian similarity matrix if and only if there does not exist a  weighted asteroidal triple in $A$. % i.e., if $A$  is $\WAT$-free.
\end{theorem}

If we apply this result to the adjacency matrix $A_G$ of a graph $G$ we obtain that $G$ is a unit interval graph if and only if there does not exist a weighted asteroidal triple in $A_G$. 
As we will show in Section \ref{secuig},  the structural characterization of unit interval graphs from Theorem \ref{thm:uig} can in fact be derived from our main result in Theorem \ref{thm:main}.
Indeed, we will show that  the notion of weighted asteroidal triple in $A_G$ subsumes the notions of claw, chordless cycle and  asteroidal triple in $G$.

%\medskip
%\tcolblue{
%We conclude with some hints on the proof technique for Theorem \ref{thm:main}. The proof will in fact be algorithmic and, given a matrix $A$, it will find either a Robinson ordering of $A$ or a weighted asteroidal triple. A first key ingredient is the notion of `homogeneous set' which (roughly) is a subset $S$ of the set $V$ indexing  $A$ having the property that alSome key ingredients that we will use in the proof are the notions of `critical element' and `homogeneous set' for a given }

\subsection{Organization of the paper}

In Section \ref{secprel} we group preliminary notions and  results that will be used in the rest of the paper, in particular, about  homogeneous sets, critical elements,  and similarity layer structures.  
Section \ref{secmain} is devoted to the proof of  our main result (Theorem~\ref{thm:main}). 
In Section \ref{secfinal}, we first show how to recover the known charaterization of unit interval graphs (Theorem~\ref{thm:uig}) from our main result. 
Then we give a simple algorithm for finding all weighted asteroidal triples (or decide that none exists), that runs in $O(n^3)$. Finally we conclude with some remarks about the problem of finding the largest Robinsonian submatrix when the given matrix is not Robinsonian.

\section{Preliminary results }\label{secprel}

In this section we introduce some notation and basic results that we will need throughout the paper.

\ignore{
\subsection{Ordered partitions}

We begin with  introducing the notion of ordered partitions.
%several  notions that we will use  throughout the paper, namely ordered partitions and  strongly homogeneous sets.
An {\em ordered partition} $\psi=(X_1,\dots, X_k)$ of $V$ is an ordered list of mutually disjoint subsets of $V$
 that cover $V$.
An ordered partition $\psi$  defines a partial order $\preceq_{\psi}$ on $V$ such that 
 $x\preceq_{\psi} y$ if and only if $x\in X_i$ and $y\in X_j$ with $i\leq j$.
 If $i=j$ then we  denote $x=_{\psi} y$ while if $i<j$ we denote $x\prec_{\psi} y$.
 When all classes $X_i$ are singletons then $\psi$ is a linear order of $V$, usualy denoted by $\sigma$.
We will introduce later in this section a special class of ordered partitions, namely similarity layer partitions of a given matrix $A$.
Given a linear order $\sigma$ and an ordered partition $\psi$ of $V$, we say that %{\em $\sigma$ extends $\psi$} (in the terminology of \cite{LP15}, 
{\em $\sigma$ is compatible with $\psi$} if, for any $x,y\in V$, $x\prec_\psi y$ implies $x\prec_\sigma y$.
}

 \subsection{Homogeneous sets and critical elements}
 
We first  introduce the notion of `homogeneous set' for  a given symmetric matrix $A$, which we then  use %to recurse and
to  reduce the problem of checking whether $A$ is Robinsonian  to the same problem on two smaller submatrices of $A$.

\begin{definition}\label{defhomoS}
Let $A$ be a symmetric matrix indexed by $V$. A set $S\subseteq V$ is said to be:
\begin{itemize}
\item {\em homogeneous} for $A$  if $A_{xy}=A_{xz}$ for all $x\in V\setminus X$ and $y,z\in X$;
\item {\em strongly homogeneous} for $A$ if $A_{xy}=A_{xz}\leq A_{yz}$ for all $x\in V\setminus X$ and $y,z\in X$;
\item {\em proper} if $2\le |S|\le |V|-1$.
\end{itemize}
%   \end{definition}
 Assume that $S$ is  proper strongly homogeneous for $A$.  We will consider the following two submatrices of $A$: 
\begin{itemize}
\item the {\em restriction $A[S]$ of $A$ to $S$}, which  is the submatrix of $A$ indexed by $S$;
\item the {\em contraction $A/S$ of $A$ by $S$}, which is defined as the submatrix $A[\overline S \cup\{s\}]$, where $\overline S=V\setminus S$ and $s$ is an arbitrary element of $S$ (thus contracting $S$ to a single element).
\end{itemize}
\end{definition}

\noindent
These definitions are motivated by the following lemma which shows that, if $S$ is a proper strongly homogeneous set for $A$, then  the problem of recognizing whether $A$ is Robinsonian  and if not of finding a weighted asteroidal triple can be reduced to the same problem for the two matrices $A[S]$ and $A\slash S$.

\begin{lemma}\label{lemreduce}
Let $A$ be a symmetric matrix  and let  $S$ be a  proper strongly homogeneous set for $A$.
Then $A$ is Robinsonian if and only if both  $A[S]$ and $A\slash S$ are Robinsonian. Moreover:
\begin{itemize}
\item[(i)]
If $\sigma_1=(x_1,\cdots,x_p)$ is a Robinson ordering of $S$ and $\sigma_2=(y_1,\cdots, y_{j-1}, s, y_{j},\cdots, y_q)$ is a Robinson ordering of $A\slash S$, then $\sigma=(y_1,\cdots, y_{j-1}, x_1,\cdots,x_p, y_{j},\cdots, y_q)$ is a Robinson ordering of $A$.
\item[(ii)]
Any  weighted asteroidal triple of $A[S]$ or of $A\slash S$   is  a weighted asteroidal triple of $A$.
\end{itemize}
\end{lemma}

\begin{proof}
Direct verification.
\end{proof}

In view of this lemma,  the core difficulty  in the proof of Theorem \ref{thm:main} is  the case when $A$ has no proper strongly homogeneous  set. The following notion of {\em critical} element will play a key role for analyzing this case.

%An element $v\in V$ is said to be {\em valid} if there is no pair $x, y$ of vertices in $V-v$ with 
%$x\overset{y}{\sim}v$ and $y\overset{x}{\sim}v$.
\begin{definition}\label{defcritical}
Let $A$ be a symmetric matrix indexed by $V$. Then $a\in V$ is said to be {\em critical for $A$} 
if $x\overset{a}{\sim}y$ holds for all distinct elements  $x,y\in V\setminus\{a\}$.
\end{definition}

Note that if $a$ is a critical element of a Robinsonian matrix $A$, then it must be an end point of \textit{any} Robinson ordering of $A$.
On the other hand, an end point of a Robinson ordering might not be critical.
In this work, we will study critical elements for arbitrary  (not necessarily Robinsonian) matrices.
The following lemma shows that  any symmetric matrix $A$ has  a critical element or
a proper strongly homogeneous  set.

\begin{lemma}\label{lem:04}
% There is a polynomial time algorithm that, for any g
Given a symmetric matrix $A$,  one can find   a  critical element   or 
a proper strongly homogeneous set for $A$.
  \end{lemma}

\begin{proof}
The proof relies on the following algorithm. Pick an arbitrary element $a\in V$ and construct a set $Z$ as follows.
\begin{itemize}
\item Initially set  $Z=V\setminus \{a\}$.
\item Repeat the following until $|Z|=1$.
\begin{description}
\item[(i)] If there exists an element $v\in V\setminus Z$ such that
${\rm argmin}\{A_{vz}: z\in Z\}\neq Z$, 
then pick any such $v$ and let  $Z\leftarrow {\rm argmin}\{A_{vz}: z\in Z\}$. 

\item[(ii)] Otherwise $Z$ is  homogeneous for $A$. 
If there exist   distinct elements $x, y, z$ with $x\in V\setminus Z$ and $y, z\in Z$ such that 
$A_{xy}=A_{xz}> A_{yz}$, then  let $Z\leftarrow Z\setminus \{z\}$. 
Otherwise  $Z$ is strongly homogeneous and output $Z$. 
\end{description}
\item If $|Z|=1$  then output the element in $Z$.
\end{itemize}
The proof of the  lemma will be complete if we can  show that, if the final set $Z$ is  a singleton set with (say) $Z=\{b\}$,  then $b$  is critical for $A$.
For this it suffices to  show:
\begin{equation}
\label{eq:lem:04-1}
\text{$a\overset{b}{\sim} v$ for any $v\in V\setminus \{a,b\}$},
\end{equation}
%(where $a$ is the element chosen at the beginning). 
since then, for any distinct $x,y\in V\setminus \{a,b\}$, we have $x\overset{b}{\sim} a \overset {b}{\sim} y$.
%$(x,a,y)$ is a path from $x$ to $y$ avoiding $b$.
We denote by $Z_i$ the set $Z$ obtained at the $i$-th step in the above algorithm. Then, 
$Z_0=V\setminus\{a\}$ and $Z_{i+1}\subsetneq Z_i$ for all $i\ge 0$, with  $Z_k=\{b\}$ at the last $k$-th iteration.
Relation (\ref{eq:lem:04-1}) follows if we can show that, for any $0\le i\le k-1$,
\begin{equation}\label{eq:lem:04-2}
\text{$a\overset{b}{\sim} v$ for any $v\in Z_i\setminus Z_{i+1}$}.
\end{equation}
We prove (\ref{eq:lem:04-2}) using induction on $i$. Let  $0\le i\le k-1$ and assume that (\ref{eq:lem:04-2}) holds for all 
$j\le i-1$ (when $i\ge 1$); we show that (\ref{eq:lem:04-2}) also holds for index $i$. For this let $v\in  Z_i\setminus Z_{i+1}$.
%Suppose that $a\overset{b}{\sim} v$ for any $v\in V\setminus Z_i$ with $i\leq k-1$.
%(This is obviously true for $i=0$.)
%We shall show $a\overset{b}{\sim} v$ for any $v\in Z_i\setminus Z_{i+1}$.

Assume first that  $Z_{i+1}$ is constructed from $Z_i$ as in  (i). Then there exists  $v_i\in V\setminus Z_i$
such that $Z_{i+1}={\rm{argmin}}\{A_{v_iz}: z\in Z_i\} \subsetneq Z_i$.
As $v\in Z_i\setminus Z_{i+1}$ and $b\in Z_k\subseteq Z_{i+1}$ we have $A_{v_iv}>A_{v_ib}$ and thus $(v_i,v)$ avoids $b$. 
If $v_i=a$ we are done. Otherwise, $v_i$ belongs to one of the sets $Z_{j-1}\setminus Z_j$ for some $0\le j\le i-1$ and thus, using the induction assumption, we can find a path from $a$ to $v_i$ avoiding $b$. Concatenating it with  $(v_i,v)$ we get a path from $a$ to $v$ avoiding $b$, i.e., $v\overset{b}{\sim} a$.

Assume now that  $Z_{i+1}$ is constructed from $Z_i$ as in  (ii). Then $Z_i$ is homogeneous  for $A$  and there exist elements 
$x\not\in Z_i$, $y\in Z_{i+1}$ and $v\in Z_{i}\setminus Z_{i+1}$ such that $Z_{i+1}=Z_i\setminus \{v\}$
and $A_{xy}=A_{xv}>A_{yv}$. As $b\in Z_{i+1}$ we have $v\ne b$.
We first claim that 
$A_{yv}\ge A_{bv}$. For this  assume  for contradiction that $A_{yv}<A_{bv}$.  As  $b,y\in Z_{i+1}$ we have $i\le k-2$. Moreover, as  $v\not\in Z_{i+1}$ and 
$A_{yv}<A_{bv}$, the set ${\rm{argmin}}\{A_{vu}: u\in Z_{i+1}\}$ does not contain  $b$ and thus  is a strict subset of $Z_{i+1}$. 
Note that $v$ is the only element in $V\setminus Z_{i+1}$ with this propery (i.e., ${\rm{argmin}}\{A_{wu}: u\in Z_{i+1}\}=Z_{i+1}$ for $w\in V\setminus (Z_{i+1}\cup \{v\})$) since $Z_i$ is homogeneous.
Hence,  at the next step we would construct $Z_{i+2}$ from $Z_{i+1}$ as in (i) and thus 
%when constructing $Z_{i+2}$, 
we would have $Z_{i+2}\subsetneq Z_{i+1}\setminus \{b\}$, a contradiction.
Therefore, $A_{xb}=A_{xy}=A_{xv}> A_{yv}\ge A_{vb}$ holds. 
Hence the path $(x,v)$ avoids $b$. 
If $x=a$ we are done. Otherwise, as $x\not\in Z_i$, $x$ lies in $Z_j\setminus Z_{j-1}$ for some $0\le j\le i-1$ and thus, by the induction assumption, there is a path from $a$ to $x$ avoiding $b$. Concatenating it with $(x,v)$ we get a path from $a$ to $v$ avoiding $b$, i.e., $v\overset{b}{\sim} a$.
\end{proof}

\subsection{Similarity layer partitions}

We begin with  the notion of ordered partition.
%several  notions that we will use  throughout the paper, namely ordered partitions and  strongly homogeneous sets.
An {\em ordered partition} $\psi=(X_1,\dots, X_k)$ of $V$ is an ordered list of mutually disjoint subsets of $V$
 that cover $V$.
Then  $\psi$  defines a partial order $\preceq_{\psi}$ on $V$ such that 
 $x\preceq_{\psi} y$ if and only if $x\in X_i$ and $y\in X_j$ with $i\leq j$.
 If $i=j$ then we  denote $x=_{\psi} y$ while if $i<j$ we denote $x\prec_{\psi} y$.
 When all classes $X_i$ are singletons then $\psi$ is a linear order of $V$, usually denoted by $\sigma$.
%We will introduce later in this section a special class of ordered partitions, namely similarity layer partitions of a given matrix $A$.

Given a linear order $\sigma$ and an ordered partition $\psi$ of $V$, we say that %{\em $\sigma$ extends $\psi$} (in the terminology of \cite{LP15}, 
{\em $\sigma$ is compatible with $\psi$} if, for any $x,y\in V$, $x\prec_\psi y$ implies $x\prec_\sigma y$.

\medskip
A key ingredient in the proof of Theorem \ref{thm:main} is the notion of {\em similarity layer structure}, which was introduced in \cite[Section 4.2]{Laurent16} and played a crucial role there in the study of the multisweep SFS algorithm. 
Fix an element $a\in V$.  We define  subsets $X_i$ ($i\ge 0$)  of $V$ in the following iterative manner: set $X_0=\{a\}$ and for $i\ge 1$ 
\begin{equation}\label{eqlayer}
X_i=\{ y\not\in X_0\cup\cdots \cup X_{i-1}: A_{xy}\ge A_{xz} \ \forall x\in X_0\cup\cdots \cup X_{i-1}, \forall z\not\in X_0\cup\cdots \cup X_{i-1}\}.
\end{equation}
%The sets $X_i$ are called the {\em similarity layers of $A$ rooted at $a$.} 
We let $k$ denote the largest integer for which $X_k\ne \emptyset$. The sets $X_0, \cdots, X_k$ are called the {\em similarity layers of $A$ rooted at $a$}  and we denote by 
$\psi_a=(X_0,X_1,\cdots,X_k)$ the ordered collection of the similarity layers and call it the {\em similarity layer structure} rooted at $a$. As we will see in Lemma \ref{lem:02} below, when $A$ is Robinsonian and $a$ is critical for $A$, $\psi_a$ is an ordered partition of $V$.
The following basic properties of the similarity layers follow easily from their definition. 

\begin{lemma}\label{lemlayerbasic}
The following properties hold for   $\psi_a=(X_0, X_1,\cdots,X_k)$,  the similarity layer structure of $A$ rooted at $a$:
\begin{description}
\item[(L1)]
If $x\in X_i$, $y,z\in X_j$ with $0\le i<j\le k$, then $A_{xy}=A_{xz}$.
\item[(L2)]
If $x\in X_i$, $y\in X_j$, $z\in X_h$ with $0\le i<j<h\le k$, then
$A_{xz}\le A_{xy}$.
\item[(L3)]
%If $x\in X_i$ and $y\in X_j$ with $1\le i<j\le k$,  then  $a\overset{y}{\sim} x$.
If $x\in X_i$ with $0\le i\le k$ and $y\in V\setminus (X_0\cup\cdots \cup X_i)$, then  $a\overset{y}{\sim} x$.
\item[(L4)]
%Let $k$ be the largest integer for which $X_k\ne \emptyset$.
Assume $V=X_0\cup \cdots \cup X_k$. Then the set   $X_k$ is  homogeneous for $A$. Moreover, if $|X_k|=1$ with  (say)  $X_k=\{b\}$, then $b$ is  critical for $A$.
\end{description}
\end{lemma}

In the above lemma no assumption is made on the root of the similarity layer structure. We now group further properties that hold when the root $a$ is assumed to be critical.

\begin{lemma}\label{lem:02}
Assume that $a$ is critical for $A$ and let $\psi_a=(X_0,X_1,\cdots,X_k)$ be the ordered collection of similarity layers of $A$ rooted at $a$. 
If $X_0\cup X_1\cup \cdots \cup X_k\ne V$ then we can find a weighted asteroidal triple of $A$.
\end{lemma}

\begin{proof}
Assume that $U:= X_0\cup X_1\cup \cdots \cup X_k \ne V$.
By assumption, $X_{k+1}=\emptyset$. 
We use the following notation: for $x\in U$ set $M_x:= {{\rm argmax}}\{A_{xv}: v\in V\setminus U\}$.
We claim  that there exist elements $x\ne x'\in U$ such that $M_x\setminus M_{x'}\ne \emptyset$ and $M_{x'}\setminus M_x\ne \emptyset$.
For, if not, then $M_{x_1}\subseteq \cdots \subseteq M_{x_p}$ for some ordering of the elements of 
$U=\{x_1,\cdots,x_p\}$ and thus $X_{k+1}=\bigcap_{x\in U} M_x=M_{x_1}$, contradicting the assumption  $X_{k+1}=\emptyset$.
\ignore{
As $U\ne V$ there exists $z_1\in V\setminus U$. 
 As $z_1\not\in X_{k+1}$ there exist $x_1\in U$ and $z_2\in V\setminus U$ such that $A_{x_1z_1}<A_{x_1z_2}$.
Moreover we can choose $z_2\in M_{x_1}$ and we have   $z_1\not\in M_{x_1}$.
Next, as $z_2\not\in X_{k+1}$ there exist $x_2\in U$ and $z_3\in V\setminus U$ such that $A_{x_2z_2}<A_{x_2z_3}$. Again we can choose 
$z_3\in M_{x_2}$ and we have $z_2\not\in M_{x_2}$.
Iterating we obtain elements $x_1,x_2, \cdots x_i,\cdots$ of $U$ and elements $z_1,z_2,\cdots z_i,\cdots $ of $V\setminus U$ such that 
$z_{i+1}\in M_{x_i}$ and $z_i\not\in M_{x_i}$.
Hence $z_{i+1}\in M_{x_i}\setminus M_{x_{i+1}}$ for all $i\ge 1$.  
We claim that $M_{x_{i+1}}\setminus M_{x_i}\ne \emptyset$ for some $i$. For otherwise we would have $M_{x_{i+1}}\subsetneq M_{x_i}$ for all $i$, implying a contradiction.
So we have shown that there exist elements $x\ne x'\in U$ such that $M_x\setminus M_{x'}\ne \emptyset$ and $M_{x'}\setminus M_x\ne \emptyset$.
}
Let  $u\in M_x\setminus M_{x'}$ and $v\in M_{x'}\setminus M_x$.
Then $A_{xu}>A_{xv}$ implying $x\overset{v}{\sim} u$, and
$A_{x'v}>A_{x'u}$, implying $x'\overset{u}{\sim} v$.
Combining with $a\overset{v}{\sim} x$, and $a\overset{u}{\sim} x'$ obtained from (L3) since $x,x'\in U$ and $u,v\not\in U$, we obtain 
$a\overset{v}{\sim} u$ and $a\overset{u}{\sim} v$.
Finally, $u\overset{a}{\sim} v$ since $a$ is critical, so that $\{a,u,v\}$ is a weighted asteroidal triple of $A$.
\end{proof}

\begin{lemma}\label{lem:03}
Assume that $a$ is critical for $A$ and let $\psi_a=(X_0,X_1,\cdots,X_k)$ be the ordered collection of similarity layers of $A$ rooted at $a$. 
Consider the properties:
\begin{description}
\item[(L1*)] For all  $x\in X_i$ and  $y,z\in X_j$ with $0\le i<j\le k$,  we have $A_{xy}=A_{xz}\le A_{yz}$. %(implying $A_{xz}\le \min\{A_{xy},A_{yz}\}$).
\item[(L2*)] For all  $x\in X_i$, $y\in X_j$, $z\in X_h$ with $0\le i<j<h\le k$, 
$A_{xz}\le \min\{A_{xy},A_{yz}\}$.
\end{description}
The following holds:
\begin{description}
\item[(i)] If (L1*)  or (L2*) does not hold then we can find a weighted asteroidal triple in $A$.
\item[(ii)]
If (L1*) holds then the last layer $X_k$ is strongly homogeneous.
%, or it consists of a single element which is  critical for $A$.
\end{description}
\end{lemma}

\begin{proof}
(i) Assume (L1*) does not hold, i.e., in view of (L1), there exist $x\in X_i$, $y\ne z\in X_j$ with $0\le i<j\le k$ and $A_{xy}=A_{xz}>A_{yz}$.
Then $x\overset{z}{\sim} y$ and $x\overset{y}{\sim} z$. Combining with  $a\overset{z}{\sim} x$ and $a\overset{y}{\sim} x$ from (L3) we deduce that $a\overset{z}{\sim} y$  and $a\overset{y}{\sim} z$. Finally, $y\overset{a}{\sim} z$ since $a$ is critical and thus   $\{a,y,z\}$ is a weighted asteroidal triple.
Assume now that (L2*) does not hold, i.e., in view of (L2), there exist $x\in X_i$, $y\in X_j$, $z\in X_h$ with $0\le i<j<h\le k$ and 
$A_{yz}<A_{xz}\le A_{xy}$. Then, again we have  $x\overset{z}{\sim} y$ and $x\overset{y}{\sim} z$ and thus $\{a,y,z\}$ is a  weighetd asteroidal triple for $A$.
Finally,  (ii)  follows directly from (L4) and (L1*).
 \end{proof}

\section{Structural characterization}\label{secmain}

\subsection{Main result}
We can now formulate the main technical result of this paper, from which our main Theorem \ref{thm:main} will easily follow. The  proof of Theorem \ref{thm:1} will be given in the next section. % \ref{sec:proof}.
%Our main theorem easily follows from the following.
%\begin{lemma}
%\label{lem:core}
%Let $A$ be  a symmetric matrix indexed by $V$ and  $\psi_a$ be the similarity layer rooted at a critical element $a\in V$.
%Then there is a polynomial time algorithm that  finds an AT in $A$ or 
%a Robinson ordering $\prec$ that extends $\prec_{\psi_a}$.  
%\end{lemma}
\begin{theorem}
\label{thm:1}
Let $A$ be a symmetric matrix indexed by $V$ and let $a$ be a critical element for $A$.
% and $\psi_a$ be the similarity layer rooted at a critical element $a\in V$. 
%There is a polynomial time algorithm that finds a strongly homogeneous proper set, 
Then one can find %(in polynomial time) 
one of the following three objects:
\begin{description}
\item[(i)] a proper strongly homogeneous set;
\item[(ii)] a weighted asteroidal triple;
\item[(iii)] a Robinson ordering of $A$ compatible with the similarity layer structure $\psi_a$ of $A$ rooted at $a$.
\end{description}
%an $A$-asteroidal triple, or a Robinson ordering $\prec$ that extends $\prec_{\psi_a}$.
\end{theorem}

As an application we obtain the following result, which directly implies Theorem~\ref{thm:main}.
\begin{corollary} \label{cor:main}
Let $A$ be a symmetric matrix indexed by $V$.
Then one can find %(in polynomial time) 
either a weighted asteroidal triple, or a Robinson ordering of $A$.
\end{corollary}

\begin{proof}
The proof is by induction on the size $|V|$ of $A$.  In view of  Lemma \ref{lem:04}, one can find   either a critical element $a$, or a proper strongly homogeneous set  for $A$. 
If we have  a critical element $a$,  then we can apply Theorem \ref{thm:1} to $(A,a)$ and find either a proper strongly homogeneous set, or a weighted asteroidal triple, or a Robinson ordering compatible with $\psi_a$. In the latter two cases we obtain the desired conclusion.
So we now only need to consider the case when a proper strongly homogeneous set $S$ has been obtained in one of the above steps.
% Clearly if  case (ii) or (iii) applies then we get the desired conclusion.
%So let us assume that  case (i) applies and we find a proper strongly homogeneous set $S$.
 Then we consider the two matrices $A[S]$ and $A\slash S$ which have smaller size than $A$ since $S$ is proper. By the induction assumption applied to both $A[S]$ and $A\slash S$, either we find a weighted asteroidal triple in $A[S]$ or in $A\slash S$, in which case we alse have an asteroidal triple in $A$ (by Lemma \ref{lemreduce} (ii)), or we find  Robinson orderings of $A[S]$ and $A\slash S$, which we can then combine to get a Robinson ordering of $A$ (by Lemma~\ref{lemreduce}~(i)). %This completes the proof.
\end{proof}

\subsection{Proof of Theorem \ref{thm:1}}\label{sec:proof}

This section is devoted to the proof of Theorem \ref{thm:1}. Let $A$ be a symmetric matrix indexed by $V$ and let $a\in V$ be a critical element for $A$. The proof is by induction on the size $|V|$ of  $A$. Moreover, it is algorithmic. It will go through a number of steps where,  either we stop and return a proper strongly homogeneous set or a weighted asteroidal triple, or we end up with constructing a Robinson ordering compatible with the similarity layer structure $\psi_a$ rooted at $a$.

\medskip
%First, applying Lemma \ref{lem:04}, either we  find a proper strongly homogeneous set in which case we return it and stop, or we find a critical point. So we now assume that we have found a critical point $a$ of $A$. Next 
We start with computing the similarity layer structure $\psi_a=(X_0,\cdots, X_k)$ rooted at $a$. If $X_0\cup \cdots \cup X_k\neq V$ or if property (L1*) or (L2*) does not hold, then we can find a weighted asteroidal triple (by Lemma \ref{lem:02} and Lemma \ref{lem:03} (i)) and we are done.
Hence we now assume that $X_0\cup\cdots \cup X_k=V$ and that (L1*) and (L2*) hold for $\psi_a$. If $|X_k|\ge 2$ then $X_k$ is proper strongly homogeneous
(by Lemma \ref{lem:03} (ii)) and we are done.
Hence we now assume that $|X_k|=1$, say $X_k=\{b\}$, and, in view of property (L4), we know that $b$ too is critical for $A$.

We can repeat the above reasoning to the similarity layer structure $\psi_b=(Y_0,\cdots, Y_{\ell})$ rooted at $b$.
Hence we may now also assume that $Y_0\cup\cdots \cup Y_{\ell}=V$, (L1*) and (L2*) hold for $\psi_b$, and $|Y_{\ell}|=1$. 

%Next we check if the two similarity layer partitions $\psi_a$ and $\psi_b$ are \textit{compatible}, which will imply in particular that the last layer of $\psi_b$ is $Y_{\ell}=\{a\}$.
Next we check if the  similarity layer structure$\psi_a$ is compatible with the reverse of the similarity layer structure $\psi_b$, which will imply  in particular that  the last layer of $\psi_b$ is $Y_{\ell}=\{a\}$.

\begin{claim}
\label{claim:consistency}
If there exist two distinct elements $x, y\in V$ with 
$x\prec_{\psi_a} y$ and $x\prec_{\psi_b} y$, then 
one can find a weighted asteroidal triple of $A$.
\end{claim}

\begin{proof}
Assume $a\preceq_{\psi_a} x\prec_{\psi_a} y \preceq_{\psi_a} b$ and $b\preceq_{\psi_b}  x\prec_{\psi_b} y$, so $y\ne b$. Assume first $x=a$. 
Then, using (L3) applied to $\psi_a$ and $\psi_b$ we get, respectively,  $a\overset{b}{\sim} y$ and $b\overset{y}{\sim} a$. As $a$ is critical we also have 
$b\overset{a}{\sim} y$ and thus $\{a,b,y\}$ is a weighted asteroidal triple.

Assume now $x\ne a$. Using again (L3) applied to $\psi_a$ and $\psi_b$, we get $a\overset{y}{\sim} x$ and 
$b\overset{y}{\sim} x$, implying $a\overset{y}{\sim} b$. Moreover, $b\overset{a}{\sim} y$ since $a$ is critical,  and $a\overset{b}{\sim} y$ since 
$b$ is critical. Hence $\{a,b,y\}$ is again a weighted asteroidal triple.
\end{proof}

To recap, from now on we will assume that $\psi_a=(X_0,\cdots, X_{k})$ is compatible with the reverse of  $\psi_b=(Y_0,\cdots, Y_{\ell})$, i.e., there do not exist  $x, y\in V$ with $x\prec_{\psi_a} y$ and $x\prec_{\psi_b} y$, and therefore $X_{0} = Y_{\ell}=\{a\}$ and $X_{k} = Y_{0}=\{b\}$.
We show the shape of the similarity layer partitions $\psi_a$ and $\psi_b$ in Figure \ref{fig:proof}, where the similarity layers $X_i$ and $Y_j$ are indicated by ellipses and rectangles, respectively. We also indicate the set $X_{k-1}\cap Y_{j^*}$, where $j^*$ is the largest integer $j\ge 1$ for which $X_{k-1}\cap Y_j\neq \emptyset$, which will play a crucial role in the rest of the proof.

\begin{claim}\label{claim:core5}
Let $j^*$ be the largest integer $j\ge 1$ such that $X_{k-1}\cap Y_{j}\ne \emptyset$.
If $|X_{k-1}\cap Y_{j^*}|\ge 2$ then $X_{k-1}\cap Y_{j^*}$ is proper strongly homogeneous for $A$.
\end{claim}

\ignore{
 We also  remark the following:
\begin{equation}
\label{eq:core5}
\begin{split}
&\text{For each $X_i\in \psi_a$, let $j'$ be the largest integer $j$ such that $X_i\cap Y_j\neq \emptyset$.} \\
&\text{Then $X_i\cap Y_{j'}$ is strongly homogeneous.}
\end{split}
\end{equation}
}
\begin{proof}
For this pick $x\not \in X_{k-1}\cap Y_{j^*}$ and distinct elements $y, z\in X_{k-1}\cap Y_{j^*}$.
If $x$ lies in $X_{k-1}\cup X_k$ then $x\in Y_j$ for some $0\le j<j^*$ and thus $A_{xy}=A_{xz}\le A_{yz}$ follows from property (L1*) applied to $\psi_b$.
Otherwise $x$ lies in some $X_i$ with $i\le k-2$ and $A_{xy}=A_{xz}\le A_{yz}$ follows from property (L1*) applied to $\psi_a$.
\end{proof}
%Indeed, for any  $x\in V\setminus (X_i\cap Y_{j'})$ and $y,z\in X_i\cap Y_{j'}$, 
%we have either $x\prec_{\psi_a} y=_{\psi_a} z$ or $x\prec_{\psi_b} y=_{\psi_b} z$ by (\ref{eq:core1}),
%and hence $A_{xy}=A_{xz}\leq A_{yz}$ by (L1) of $\psi_a$ and $\psi_b$.
%Thus (\ref{eq:core5}) holds
%Suppose now that $|Y_{\ell}|\geq 2$ (where $Y_{\ell}$ denotes the last component of $\psi_b$).

%Let $j^*$ be the largest index satisfying $X_{k-1}\cap Y_{j^*}\neq \emptyset$. 
%By (\ref{eq:core1}),  $(X_0, X_1, \dots, X_{k-2}, X_{k-1}\cap Y_{j^*}, Y_{j^*-1}, \dots, Y_1, Y_0)$ is an ordered partition of $V$.
%Also  $X_{k-1}\cap Y_{j^*}$ is strongly homogeneous by (\ref{eq:core5}),
From now on we  assume that   $|X_{k-1}\cap Y_{j^*}|=1$ and we  set $X_{k-1}\cap Y_{j^*}=\{c\}$.
Thus we may partition the set $V$ as
$$V= \underbrace{X_0\cup \cdots \cup X_{k-2}}_{=: X} \cup \{c\}\cup 
\underbrace{ Y_{j^*-1}\cup \cdots \cup Y_0}_{=: Y} = X\cup\{c\} \cup Y.$$
%Let $X=\bigcup_{i=0}^{k-2} X_i$ and $Y=\bigcup_{j=0}^{j^*-1} Y_j$.

\begin{figure}
\begin{center}
\resizebox{\textwidth}{!}{
\includegraphics[scale=1]{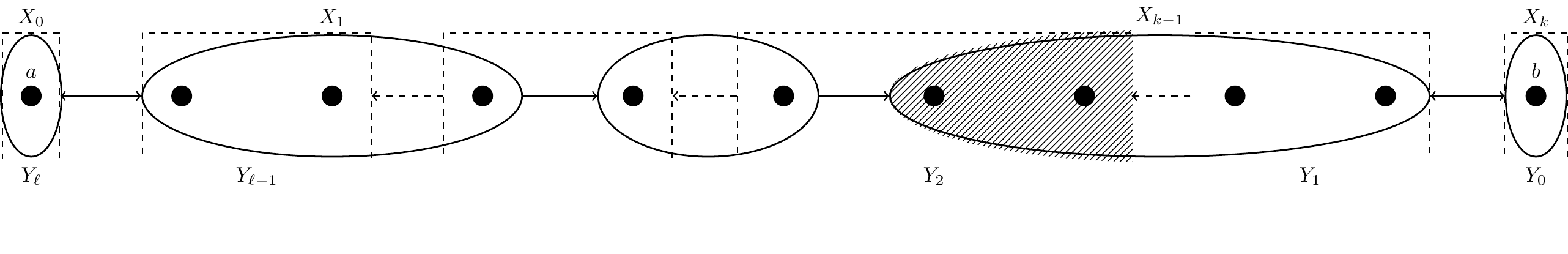}
}
\end{center}
%\caption{Representation of the subset $X_{k-1} \cap Y_{j^*}$ for $j^*=2$.}
\caption{The siimilarity layer structures $\psi_a$, $\psi_b$ and the set $X_{k-1} \cap Y_{j^*}$ for $j^*=2$.}
\label{fig:proof}
\end{figure}

For further use we record the following consequence of  (L3) applied to $\psi_a$ and $\psi_b$:
\begin{equation}
\label{eq:core7}
\text{For any $u, v\in X$ (resp.,~$u, v\in Y$),\  $u\overset{c}{\sim} v$ in $A$.}
\end{equation}

At this step we now need to work with two new matrices $A^X$ and $A^Y$ that are indexed, respectively, by $X\cup\{c\}$ and $Y\cup \{c\}$ and constructed by modifying the entries of $A$ in the following way.
Let $M$ be a positive integer, chosen sufficiently large, so that   
\begin{equation}\label{eqM}
M> 2\max\{|A_{uv}|: u,v\in V\}.
\end{equation}
Let $A^X$ be the symmetric matrix indexed by $X\cup\{c\}$,  obtained from $A[X]$ by adjoining a new column/row indexed by $c$  with entries:
\begin{equation}\label{eq:AX}
A^X_{cv}=-M-j+\frac{A_{cv}}{M} \quad \text{ for }v\in X \text{ with }  v\in Y_j.
\end{equation}
Similarly, let $A^Y$ be the symmetric matrix indexed by $Y\cup\{c\}$, obtained from $A[Y]$ by adding a new column/row indexed by $c$   with entries: 
\begin{equation}\label{eq:AY}
A^Y_{cv}=-M-i+\frac{A_{cv}}{M} \quad \text{ for }v\in Y \text{ with } v\in X_i.
\end{equation}
Note that $j^*\le j\le \ell$ in (\ref{eq:AX}) and $  k-1\le i\le k$ in (\ref{eq:AY}).
%With this definition we have 
%\begin{align}
%&A_{uv}^X > \max\{A_{u\tilde{c}}^X, A_{v\tilde{c}}^X\} \text{ for any $u, v\in X$} \\
%&\text{If $u\prec_{\psi_b} v$, $A_{\tilde{c}u}>A_{\tilde{c}v}$} \\
%&\text{If $u=_{\psi_b} v$, then $A_{\tilde{c}u}-A_{\tilde{c}v}=A_{cu}-A_{cv}$}
%\end{align}

\begin{claim}\label{claim:core3}
\begin{description}
\item[(i)] 
The element $a$ is critical in $A^X$ and the similarity layer structure of $A^X$ rooted at $a$, denoted as $\psi^X_a$,  is equal to 
$\psi^X_a=(\{a\}, X_1,\cdots, X_{k-2},\{c\})$.
\item[(ii)] The element $b$ is critical in $A^Y$ and the similarity layer structure of $A^Y$ rooted at $b$, denoted as $\psi^Y_b$,  is equal to 
$\psi^Y_b=(\{b\}, Y_1,\cdots, Y_{j^*-1},\{c\})$.
\end{description}
\end{claim}

%With this definition, we first show that 
\ignore{
\begin{equation}
\label{eq:core3}
\text{$a$ (resp.~$b$) is critical in $A^X$ (resp.~$A^Y$).}
\end{equation}
}

\begin{proof}
(i) We show that $a$ is critical for $A^X$. For  any $v\in X\setminus \{a\}$, note first using definition (\ref{eq:AX}) that $A^X_{vc}>A^X_{ac}$ holds, since
\[
A^X_{vc}-A^X_{ac}=-j+\ell+\frac{A_{vc}-A_{ac}}{M} \geq 1+\frac{A_{vc}-A_{ac}}{M}>0,
\]
where the first two relations  follow from  $v\in Y_j$ and $a\in Y_{\ell}$ with $j\le \ell-1$ and the third inequality follows from (\ref{eqM}). 
 Hence, the path $(v,c)$ avoids $a$ in $A^X$. 
Now, for $x\ne y\in X\setminus \{a\}$, the path $(x,c,y)$ avoids $a$ in $A^X$, which shows that $a$ is critical for $A^X$.
Moreover, as $A^X_{vc}<A^X_{xy}$ for all $v,x,y\in X$, it follows that the similarity layer structure of $A^X$ rooted at $a$ has indeed the desired form.

%The proof that $b$ is $A^Y$-critical is analogous using  definition (\ref{eq:AY}).
The proof of (ii) is analogous.
 \end{proof}

%$\psi_X$ (resp.~$\psi_Y$) is the similarity layer rooted at $a$ (resp.~$b$) in $A^X$ (resp.~$A^Y$).
%Now by (\ref{eq:core3}) one can compute the similarity layer rooted at $a$ (resp.~at $b$) in $A^X$ (resp.~in $A^Y$). 
%However, since the $A^X$-value on an edge incident to $c$ is sufficiently smaller than that on the existing edges in $A$,  
%$\psi_X:=(X_0,\dots, X_{k-2}, \{c\}\})$ (resp.~$\psi_Y:=(Y_0, \dots, Y_{j^*-1}, \{c\})$) 
%is the similarity layer rooted at $a$ (resp.~$b$) in $A^X$ (resp.~$A^Y$).

Since the size of  both matrices  $A^X$ and $A^Y$ is smaller than that of $A$, we can apply the induction assumption to $(A^X,a)$ and $(A^Y,b)$, which gives the following three cases:\\
{\bf Case 1:}  we find a proper strongly homogeneous  set in $A^X$ or in $A^Y$;\\
{\bf Case 2:}  we find  a weighted asteroidal triple in  $A^X$ or  in $A^Y$; \\
{\bf Case 3:} or we  find   Robinson orderings  of $A^X$ and  $A^Y$ that are compatible with $\psi^X_a$ and $\psi^Y_b$, respectively.\\
%extend $\prec_{\psi^X_a}$ and $\prec^{\psi_Y}_b$, respectively.
We now deal with each of these three cases separately.

\medskip
\noindent
{\bf Case 1:} We assume  that we have found a proper strongly homogeneous  set $S$ in $A^X$. (The case of  $A^Y$ is similar and thus omitted.)
%The algorithm tries to check whether $S$ is strongly homogeneous or not based on the following.
As we now show,  either we can claim that $S$ is strongly homogeneous in $A$, or we find a weighted asteroidal triple in $A$.

\begin{claim}
\label{claim:case1}
Let $S$ be a proper strongly homogeneous  set in $A^X$.
%Then $c\notin S$ and $S$ is homogeneous.
Then, either $S$ is strongly homogeneous in $A$, or 
%$S\subseteq Y_{j^*}$ and
 there exist $x\ne x'\in S$ such that 
 %$A_{cx}=A_{cx'}>A_{xx'}$ in which case  
$\{x,x',c\}$ is  a weighted asteroidal triple for $A$.

% if and only if $A_{cx}\leq A_{xx'}$ for any $x, x'\in Y_{j^*}\cap S$.
%(Thus $S$ is always strongly homogeneous if $Y_{j^*}\cap S=\emptyset$.)
\end{claim}

\begin{proof}
Let $S\subseteq X\cup\{c\}$ be a proper strongly homogeneous set in $A^X$.
We first show  $c\not\in S$. For this, suppose for contradiction that  $c\in S$.
Since $S$ is proper, we can take elements $x\in S\setminus \{c\}$ and $v\in X\setminus S$.
Since $S$ is homogeneous in $A^X$, we have $A^X_{vc}=A^X_{vx}$, which gives 
$-M-j+\frac{A_{vc}}{M}=A_{vx}$ if $v\in Y_j$.  This however contradicts the choice of $M$ in (\ref{eqM}).
Therefore, $c\notin S$.
%\begin{equation}
%\label{eq:case1}
%c\notin S.
%\end{equation}

%We next show 
%\begin{equation}
%\label{eq:case1}
%\text{$S$ is contained in a component in $\psi_b$ and $S$ is homogeneous in $A$.}
%\end{equation}
Take any $x, x'\in S$. As $c\notin S$ and $S$ is homogeneous in $A^X$, we have 
$A_{cx}^X=A_{cx'}^X$, which gives 
 $-M-j+\frac{A_{cx}}{M}=-M-j'+\frac{A_{cx'}}{M}$, where $x\in Y_j$ and $x'\in Y_{j'}$.
Using again (\ref{eqM}) we derive that  $j=j'$ and 
\begin{equation}
\label{eq:case1-2}
A_{cx}=A_{cx'} \quad \text{ for all } x, x'\in S.
\end{equation}
Therefore, $S$ is contained in some layer $Y_j$ of $\psi_b$ for some  $j\ge j^*$. Moreover,
\begin{equation}
\label{eq:case1-2b}
\text{ if } j>j^*\  \ \text{ then } \ A_{cx}=A_{cx'} \le A_{xx'} \quad \text{ for all } x, x'\in S,
\end{equation}
which follows from (L1*) applied to $\psi_b$ (since $c\in Y_{j^*}$ and $x,x'\in Y_j$ with $j>j^*$). 
Next we  claim that
\begin{equation}
\label{eq:case1-3}
A_{vx}=A_{vx'}\leq A_{xx'} \quad \text{ for all } x, x'\in S, \ v\in V\setminus (S\cup \{c\}).
\end{equation}
If $v\in X$, then $A_{vx}=A_{vx'}\leq A_{xx'}$ follows from $A_{vx}^X=A_{vx'}^X\leq A_{xx'}^X$ 
since $A$ and $A^X$ coincide on the triple $\{v, x, x'\}\subseteq X$. 
If $v\notin X$, then $v\in Y$ and thus $v\in Y_{h}$ for some $h<j$, in which case 
$A_{vx}=A_{vx'}\leq A_{xx'}$ follows from (L1*) applied to $\psi_b$ since $x,x'\in Y_j$.
%hence $v\prec_{\psi_b} x=_{\psi_b}x'$.
%Since $S$ is contained in a component in $\psi_b$, (L1) of $\psi_b$ gives (\ref{eq:case1-3}).
%(\ref{eq:case1-2}) and (\ref{eq:case1-3}) imply that $S$ is always homogeneous, 
%and what is needed for $A$ to be strongly homogeneous is $A_{cx}\leq A_{xx'}$ for each $x, x'\in S$.
%Checking this condition is simplified as in the statement of the claim 
%\begin{equation}
%\label{eq:case1-4}
%A_{cx}=A_{cx'}\leq A_{xx'} \text{ if $x, x'\notin Y_{j^*}$}
%\end{equation}

Hence, in view of  relations (\ref{eq:case1-2b}) and (\ref{eq:case1-3}), if  the set $S$ is not strongly homogeneous in $A$, then necessarily
$j=j^*$ and there exist $x\ne x'\in S$ such that $A_{cx}=A_{cx'}>A_{xx'}$.
Then $c\overset{x}{\sim} x'$ and $c\overset{x'}{\sim} x$. On the other hand, by (L3) applied to $\psi_a$, we have 
%$x\prec_{\psi_a} c$ and $x'\prec_{\psi_a} c$, we have 
$x\overset{c}{\sim} a\overset{c}{\sim} x'$ (since $c\in X_{k-1}$, $x\in X_i$, $x'\in X_{i'}$ with $i,i'<k-1$).
Thus $\{x, x', c\}$ forms a weighted asteroidal triple in $A$.
\end{proof}
%The symmetric argument can be applied to solve the case when we found a strongly homogeneous proper set in $A^Y$. 

\medskip
\noindent
{\bf Case 2:} We now assume that we have found a weighted asteroidal triple in $A^X$. (The case of having a weighted asteroidal triple in $A^Y$ is similar and thus omitted.) Our goal is to construct from it a weighted asteroidal triple in $A$. 
For this we  use the following claim.

\begin{claim}
\label{claim:constructing}
Given $u,v\in X$ and a path from $u$ to $c$ avoiding $v$ in $A^X$, the following holds:
%one can perform the following (in linear time): 
\begin{description}
\item[(i)] If $v\neq_{\psi_b} c$, one can find a path from $u$ to $b$ avoiding $v$ in $A$;
\item[(ii)] If $v=_{\psi_b} c$, one can find a path from $u$ to $c$ avoiding $v$ in $A$;
\item[(iii)] If $v\prec_{\psi_b} u$, one can find a weighted asteroidal triple in $A$.
\end{description}
\end{claim}

\begin{proof}
%Note that, since $u,v\in X$, $u\preceq_{\psi_b} c$ and $v\preceq_{\psi_b} c$. 
Say, $P=(u,\cdots, u', c)$ is a path from $u$ to $c$ avoiding $v$ in $A^X$, with $(u',c)$ as its  last edge. Note that $P'=(u,\cdots, u')$ is  a path  avoiding $v$ in $A$ and thus $u\overset{v}{\sim} u'$ not only in $A^X$ but also in $A$. 
We claim:
\begin{equation}
\label{eq:claim:constructing}
\text{either $u'\prec_{\psi_b} v$,\  or $u'=_{\psi_b} v$ and $A_{u'c}>A_{vc}$.}
\end{equation}
Indeed, since $(u',c)$ avoids $v$ in $A^X$, we have $A^X_{u'c}>\min \{A^X_{u'v}, A^X_{vc}\}=A^X_{vc}$,
where the last equality follows from the definition of $A^X$ and the choice of $M$.
Say $v\in Y_j$ and $u'\in Y_{j'}$. Then, $A^X_{u'c}>A^X_{vc}$ implies 
$j-j'> {A_{vc}-A_{u'c}\over M}$. We cannot have $j'>j$ since then one would have 
% $A_{cv}\ge A_{cu'}$ by the definition of the layers in $\psi_b$, so that 
%$0>j-j'> {A_{vc}-A_{u'c}\over M}\ge 0$, 
\[
0<j-j'-\frac{A_{vc}-A_{u'c}}{M}<-1-\frac{A_{vc}-A_{u'c}}{M}<0,
\]
reaching a contradiction.  
Hence, either 
$j'<j$ (i.e., $u'\prec_{\psi_b} v$), or $j=j'$ (i.e., $v=_{\psi_b} u'$) and $A_{vc}<A_{u'c}$.

\smallskip \noindent (i)
Assume $v\neq_{\psi_b} c$. %, i.e., $v\in Y_j$ for some $j>j^*$. 
Then,  $c\prec_{\psi_b} v$ and, by (L3) applied to $\psi_b$,  $b\overset{v}{\sim} c$ in $A$. 
If $u'\prec_{\psi_b} v$, then $u'\overset{v}{\sim} b$ in $A$ (again by (L3) applied to $\psi_b$) and thus   $u\overset{v}{\sim} u' \overset{v}{\sim} b$ in $A$, giving the desired conclusion.
Otherwise, in view of  (\ref{eq:claim:constructing}), $u'=_{\psi_b} v$ %i.e., $u'\in Y_j$, 
and $A_{u'c}>A_{vc}$, which implies $u'\overset{v}{\sim} c$ in $A$. We obtain $u\overset{v}{\sim} u'\overset{v}{\sim}c \overset{v}{\sim} b$ in $A$, giving again the desired conclusion.
%By $v\neq_{\psi_b} c$ (and hence $c\prec_{\psi_b} v$ by $v\in X$), we also have $c\overset{v}{\sim} b$ in $A$ by (L3) of $\psi_b$. Thus we get $u\overset{v}{\sim} u' \overset{v}{\sim} c \overset{v}{\sim} b$ in $A$.

\smallskip \noindent(ii)
Assume $v=_{\psi_b} c$. Then $u'\prec_{\psi_b} v$ cannot occur since $c\preceq_{\psi_b} u'$.
In view of (\ref{eq:claim:constructing}),  we have $A_{u'c}>A_{vc}$, which implies   $u'\overset{v}{\sim} c$ and thus 
$u\overset{v}{\sim} u' \overset{v}{\sim} c$ in $A$, giving the desired conclusion.

\smallskip \noindent(iii)
We assume  $v\prec_{\psi_b} u$. 
We consider two cases depending whether  $v=_{\psi_b} c$ or not.
Assume first  that $v\neq_{\psi_b} c$. 
By Claim~\ref{claim:constructing}(i), we have $u\overset{v}{\sim} b$ in $A$.
Moreover, as $v\prec_{\psi_b} u$, we get $b\overset{u}{\sim}v$ in $A$ (by (L3) applied to $\psi_b$). Finally, $u\overset{b}{\sim}v$ in $A$, since $b$ is critical in $A$, and thus the triple $\{b,u,v\}$ is a weighted asteroidal triple in $A$.
Finally assume $v=_{\psi_b} c$.
Then, by Claim~\ref{claim:constructing}(ii), we have $u\overset{v}{\sim} c$ in $A$. 
Also, as $c=_{\psi_b} v\prec_{\psi_b} u$, 
we get $v\overset{u}{\sim} b\overset{u}{\sim} c$ in $A$ (applying (L3) to $\psi_b$).
Finally, using (\ref{eq:core7}), we get  $u\overset{c}{\sim} v$ in $A$.
Hence, $\{c,u,v\}$ is a weighted   asteroidal triple in $A$.
\end{proof}

As a direct application of Claim \ref{claim:constructing}, we get the following result, which we use to show Claim \ref{claim:constructing2} below.

\begin{corollary}\label{cor:constructing}
Consider distinct elements $u,v,w\in X$ and a path $P$ from $u$ to $w$ avoiding $v$ in $A^X$. If $P$ does not contain  $c$ then $P$ is also a path avoiding $v$ in $A$. If $P$ contains $c$ then one can construct from it a path $P'$ from $u$ to $w$ avoiding $v$ in $A$.
\end{corollary}

%For any distinct three elements $u,v,w\in X$, consider an $A^X$-path $P$ between $u$ and $v$ avoiding $w$.
%If $P$ does not intersect $c$, $P$ is itself an $A$-path between $u$ and $v$ avoiding $w$.
%Claim~\ref{claim:constructing} says that, even if $P$ involves $c$, $P$ can be converted to an $A$-path between $u$ and $v$ avoiding $w$ in linear time.
%With this tool in our hand, we can now show the following.
 
\begin{claim}
\label{claim:constructing2}
Given a weighted asteroidal triple $\{x,y,z\}$ in $A^X$, one can construct from it a weighted asteroidal triple in $A$.
\end{claim}
\begin{proof}
We may assume without loss of generality  that $ x \preceq_{\psi_b} y\preceq_{\psi_b} z$.
If  $\{x,y,z\}\cap \{c\}=\emptyset$ then it follows directly from Corollary \ref{cor:constructing} that $\{x,y,z\}$ is also a weighted asteroidal triple in $A$.
%$Then Claim~\ref{claim:constructing}(i)(ii) implies that, for any distinct $u,v,w$ from $\{x,y,z\}$, 
%one can convert any $A^X$-path between $u$ and $v$ avoiding $w$ to an $A$-path between $u$ and $v$ avoiding $w$ in linear time.
%Thus $\{x,y,z\}$ remains an $A$-asteroidal triple.
Without loss of generality, we now assume  that $c= x$. By Claim~\ref{claim:constructing}(iii) one can find a weighted asteroidal triple in $A$  if $y\prec_{\psi_b} z$.
Hence we now  assume that $c= x \preceq_{\psi_b} y=_{\psi_b} z$.

If $c= x =_{\psi_b} y=_{\psi_b} z$ then, by Claim~\ref{claim:constructing}(ii), we have $c \overset{y}{\sim} z$ and $c \overset{z}{\sim} y$ in $A$.
Moreover, by (\ref{eq:core7}), $y\overset{c}{\sim} z$ in $A$.
Hence  $\{c, y, z\}$ is a weighted asteroidal triple in $A$.

If $c= x \prec_{\psi_b} y=_{\psi_b} z$ then, by Claim~\ref{claim:constructing}(i), we have $b \overset{y}{\sim} z$ and $b \overset{z}{\sim} y$ in $A$.
Moreover, as  $b$ is critical in $A$, $y\overset{b}{\sim} z$ in $A$.
Hence  $\{b, y, z\}$ is a weighted asteroidal triple in $A$.
\end{proof}

%By the symmetric argument one can show the following.
%\begin{claim}
%\label{claim:constructing3}
%Given an $A^Y$-asteroidal triple $\{x,y,z\}$, one can construct an $A$-asteroidal triple in linear time.
%\end{claim}

%Thus, if we could find an $A^X$-asteroidal triple, the algorithm outputs an $A$-asteroidal triple. 
%The symmetric argument can be applied if we found an $A^Y$-asteroidal triple.

\medskip
\noindent
{\bf Case 3:}
The remaining case is when we have found a Robinson order $\sigma_X$ (resp., $\sigma_Y$)  of $A^X$ (resp., of $A^Y$),
 which is compatible with the similarity layer structure $\psi_a^X$ of $A^X$ rooted at $a$ (resp., with the similarity layer structure  $\psi_b^Y$ of $A^Y$ rooted at $b$). 
%$\prec_{\psi_X}$ (resp., $\prec_{\psi_Y}$).
By~Claim \ref{claim:core3}, we must have $\sigma_X=(a,\cdots, c)$ and $\sigma_Y=(b,\cdots, c)$. 
We define the linear order $\sigma=(\sigma_X, \sigma_Y^{-1})$ of $V$, obtained by concatenating $\sigma_X$ and the reverse of $\sigma_Y$ along the element $c$.
%We will   construct a total order $\prec$ of $V$ by simply concatenating $\prec_5$ and $\prec_6^{-1}$ at $c$, where  $\prec_6^{-1}$ is the reverse of $\prec_6$.
In view of the form of $\psi^X_a$ %and $\psi^Y_b$ 
in Claim \ref{claim:core3}, it follows that $\sigma$ is compatible with $\psi_a$.
%By  construction of $\psi_X$ and $\psi_Y$,  $\prec$ extends $\prec_{\psi_a}$.
In order to complete  the proof of Theorem \ref{thm:1} we  need to show that, either $\sigma$ is a Robinson ordering of $A$, or we can find a weighted asteroidal triple in $A$. 

Recall that  an ordered triple $(x,y,z)$ is  Robinson in  $A$  if  $A_{xz}\le \min\{A_{xy},A_{yz}\}$ holds.
We will show that for any triple $(x,y,z)$ with
$x\prec_\sigma y\prec_\sigma z$, either $(x,y,z)$ is Robinson, or $\{x,y,z\}$ is a weighted asteroidal triple for $A$.
%To see that $\prec$ is Robinson, we check (i) and (ii) in Proposition~\ref{prop:good}.
%\begin{equation}\label{eq:R}
%\text{if $x\prec_\sigma y\prec_\sigma z$ then $A_{xz}\le \min\{A_{xy},A_{yz}\}$.}
%\end{equation}

Assume $x\prec_\sigma y\prec_\sigma z$. Then,   $x\preceq_{\psi_a} y\preceq_{\psi_a} z$, since $\sigma$ is compatible with $\psi_a$. If $x\prec_{\psi_a} y\preceq_{\psi_a} z$, then we can conclude that $(x,y,z)$ is Robinson  (using (L1*)-(L2*) applied to $\psi_a$).
Hence from now on  we may assume that 
$x=_{\psi_a} y\preceq_{\psi_a} z$. In  the next two claims we will  consider separately the two cases: $x=_{\psi_a}y \prec_{\psi_a} z$ and 
$x=_{\psi_a}y=_{\psi_a} z$.

Note that  we can analogously conclude that  $(x,y,z)$ is Robinson if $z\prec_{\psi_b} y\preceq _{\psi_b} x$ (using (L1*)-(L2*) applied to $\psi_b$). 
%Hence we may also assume that $z=_{\psi_b} y\preceq _{\psi_b} x$.

\begin{claim}
\label{claim:checking2}
Consider $x,y,z\in V$  such that  $x\prec_\sigma y\prec_\sigma z$ and   $x=_{\psi_a} y\prec_{\psi_a} z$.
If the triple 
$(x,y,z)$ is not  Robinson  in  $A$ then $\{x,y,z\}$ is a weighted asteroidal triple in $A$. 
%Moreover, this can happen
%only if $c=z=_{\psi_b}y \preceq_{\psi_b} x$.
\end{claim}

\begin{proof}
%Take any distinct $x, y, z$ with $x\prec y\prec z$ and $x=_{\psi_a} y\prec_{\psi_a} z$.
Assume $x\prec_\sigma y\prec_\sigma z$,  $x=_{\psi_a} y\prec_{\psi_a} z$, and $(x,y,z)$ is not a Robinson triple in $A$. Then
$A_{xz}>\min\{A_{xy},A_{yz}\}$ and thus $x\overset{y}{\sim }z$ in $A$.
We first claim that $z\not\in X$. Indeed, assume $z\in X$. Then, as  $x=_{\psi_a} y\prec_{\psi_a} z$ we have $\{x,y,z\}\subseteq X$ and thus, as  $\sigma$ restricts to  $\sigma_X$ on $X$, 
$x\prec_{\sigma_X} y\prec_{\sigma_X} z$. As $\sigma_X$ is a Robinson ordering of $A^X$, this implies that $(x,y,z)$ is a  Robinson triple in $A^X$ and thus in $A$, a contradiction. Therefore, $z\in X_{k-1}\cup\{b\}$. 

Next we claim that 
\begin{equation}\label{eq:z}
z\prec_{\psi_b} x\prec_{\psi_b} y\ \text{ does not hold.}
\end{equation}
For this assume $z\prec_{\psi_b} x\prec_{\psi_b} y.$ As $x=_{\psi_a} y\prec_{\psi_a} z$, we have $x,y\in X_i$ for some $i\le k-1$. We first claim that  $i\le k-2$. This is clear if $z\in X_{k-1}$. Assume now $z=b$ and $x,y\in X_{k-1}$. Then $\{x,y,z\}\subseteq Y\cup\{c\}$ and, as $\sigma$ restricts to $\sigma_Y^{-1}$ on $Y\cup\{c\}$ and 
 $\sigma_Y$ is compatible with $\psi_b^Y$, $x\prec_\sigma y\prec_\sigma z$ implies $z\preceq_{\psi_b} y\preceq_{\psi_b} x$, a contradiction.  So we have shown that $x,y\in X_i$ for some $i\le k-2$. This  implies 
$x\prec_\sigma y\prec_\sigma c$ and thus $x\prec_{\sigma_X} y \prec_{\sigma_X} c$. Say $x\in Y_j$, $y\in Y_h$, with $h>j$ as $x\prec_{\psi_b} y$.
As $\sigma_X$ is a Robinson ordering of $A^X$ we can conclude
$A^X_{xc}\le \min\{A^X_{xy},A^X_{yc}\}= A^X_{yc}$, which implies $1\le h-j \le {A_{cy}-A_{cx}\over M}<1 $ (by the choice of $M$ in (\ref{eqM})), a contradiction.
So we have shown (\ref{eq:z}).

Recall that the reverse of $\psi_b$ is compatible with $\psi_a$.
Hence  it follows from $x=_{\psi_a} y\prec_{\psi_a} z$ that $z\preceq_{\psi_b}x$ and $z\preceq_{\psi_b} y$. Moreover, we claim that
\begin{equation}
\label{eq:z2}
\text{$z=_{\psi_b}x$ or $z=_{\psi_b} y$.}
\end{equation}
Indeed, if (\ref{eq:z2}) does not hold, then $z\prec_{\psi_b} x\prec_{\psi_b} y$ or $z\prec_{\psi_b} y \preceq_{\psi_b} x$.
The former does not hold by (\ref{eq:z}),  while the latter does not hold since $(x,y,z)$ is not Robinson (as observed just before Claim \ref{claim:checking2}).
Thus  (\ref{eq:z2}) holds.
Then, by $x=_{\psi_a} y\prec_{\psi_a} z$, (\ref{eq:z2}) implies  $z=c$. 

We next claim that $y\preceq _{\psi_b}x$.
For this assume for contradiction that $x\prec_{\psi_b} y$. As above, let $x\in Y_j$, $y\in Y_h$ with $h>j$. From this (and the definition of $M$ in (\ref{eqM})), it
follows that 
$A^X_{cx}>\min\{A^X_{cy},A^X_{xy}\}=A^X_{cy}$. As $\sigma_X$ is a Robinson ordering of $A^X$ we must have $y\prec_{\sigma_X} x\prec_{\sigma_X} c$, which implies
$y\prec_\sigma x\prec_\sigma c$, a contradiction.

In total we have shown  that the following relation holds:
%\begin{equation}\label{eq:z1}
$$c=z=_{\psi_b} y \preceq_{\psi_b} x.$$
%\end{equation}
We will now show that  $\{x,y,z\}$ is a weighted asteroidal triple in $A$. 
We already have $x\overset{y}{\sim} z$, since  by assumption the triple $(x,y,z)$ is not Robinson in $A$.
Moreover, as $x=_{\psi_a}y\prec_{\psi_a} z$ we have $x\overset{z}{\sim}y$. Indeed, this follows from (L3) applied to $\psi_a$: if $a\not\in \{x,y\}$ then $x\overset{z}{\sim}a\overset{z}{\sim}y$ and, if (say) $a=x$, then $x=a\overset{z}{\sim}y$.
Remains to show $y\overset{x}{\sim}z$. We distinguish two cases.
If $c=z=_{\psi_b}y\prec_{\psi_b}x$ then $y\overset{x}{\sim}b\overset{x}{\sim}z$ follows (using  (L3) applied to $\psi_b$).
Assume now  $c=z=_{\psi_b}y=_{\psi_b}x$. 
Then, we have
$\{x,y,z\}\subseteq  X\cup\{c\}$, and $x\prec_\sigma y\prec_\sigma z$ implies $x\prec_{\sigma_X} y\prec_{\sigma_X} z$.
This  in turn implies 
$A^X_{xz}\le \min\{A^X_{xy},A^X_{yz}\}=A^X_{yz}$
and thus $A_{xz}\le A_{yz}$.
Combining with  $A_{xz}>\min\{A_{xy},A_{yz}\}$ we get 
%$A_{yz}\ge A_{xy}$, 
$A_{xz}>A_{xy}$ and thus
$A_{yz}\ge A_{xz}>A_{xy}$ which gives $y\overset{x}{\sim}z$.
So we have shown that $\{x,y,z\}$ is a weighted asteroidal triple in $A$ and  this concludes the proof.
\end{proof}

%A similar argument can be used to check the following.
\begin{claim}
\label{claim:checking3}
Consider $x,y,z\in V$ such that  $x\prec_\sigma y\prec_\sigma z$ and $x=_{\psi_a} y=_{\psi_a} z$.
If $(x,y,z)$ is not a Robinson triple  in $A$ then $\{x,y,z\}$ is a weighted asteroidal triple in $A$. 
%Moreover this can happen only if  $z=_{\psi_b} y \prec_{\psi_b} x=c.$
\ignore{
Then the Robinsonian inequality for $\{x,y,z\}$ fails if and only if  $c=_{\psi_b} x \prec_{\psi_b} y =_{\psi_b} z$.

the following holds:
\begin{description}
\item[(c)] $c=_{\psi_b} x \prec_{\psi_b} y =_{\psi_b} z$ and $A_{xz}> A_{xy}$.
\end{description}
Moreover, if it fails, then $\{x,y,z\}$ forms an $A$-asteroidal triple.
}
\end{claim}

\begin{proof}
By assumption, $x,y,z\in X_i$ for some $i\le k-1$ and the triple $(x,y,z)$ is not Robinson in $A$.
We first claim that $i=k-1$. For this assume $i\le k-2$. Then $\{x,y,z\}\subseteq X$ and thus  $x\prec_\sigma y\prec_\sigma z$
implies $x\prec_{\sigma_X} y\prec_{\sigma_X} z$. As $\sigma_X$ is a Robinson ordering of $A^X$ and thus of $A[X]$ it follows that $(x,y,z)$ is Robinson in $A$, contradicting our assumption.
Hence we know that $x,y,z\in X_{k-1}$. Then  $c\preceq_\sigma x\prec_\sigma y\prec_\sigma z$, which implies
$z\prec_{\sigma_Y} y\prec_{\sigma_Y} x\preceq_{\sigma_Y} c$ and thus  the triple $(z,y,x)$ is Robinson in $A^Y$ (since $\sigma_Y$ is a Robinson ordering of $A^Y$).
If $x\neq c$ then $(z,y,x)$ is also a Robinson triple in $A$, contradicting our assumption. Therefore we must have $x=c$. In turn, this gives 
$A^Y_{xz} \le \min\{A^Y_{xy}, A^Y_{yz}\}=A^Y_{xy}$, which implies
$A_{xz}\le A_{xy}$. On the other hand, $A_{xz}>\min\{A_{xy},A_{yz}\}$, since $(x,y,z)$ is not Robinson in $A$.
This implies 
%$A_{xy}>A_{yz}$, 
$A_{xz} >A_{yz}$ and thus $A_{xy}\ge A_{xz}>A_{yz}$, so that 
$x\overset{z}{\sim}y$ and $x\overset{y}{\sim} z$.
Lastly,  $y\overset{c=x}{\sim} z$ follows from (\ref{eq:core7}) and thus we have shown that $\{x,y,z\}$ is a weighted asteroidal triple in $A$.
\end{proof}

This concludes the proof of Theorem \ref{thm:1}.

\section{Applications}\label{secfinal}

In this section we group some applications of our characterization of Robinsonian matrices in terms of weighted asteroidal triples.
First we indicate how we can derive from it the known structural characterization of unit interval graphs from Theorem \ref{thm:uig}.
As we will see, weighted asteroidal triples offer a common framework to formulate the three types of obstructions for the graph case: chordless cycles, claws and asteroidal triples.

As another application, in order to decide whether a matrix $A$ is Robinsonian it suffices to check whether it has a weighted asteroidal triple. 
 The proof of Theorem \ref{thm:1} is algorithmic and yields a polynomial time algorithm for finding a weighted asteroidal triple (if some exists), however we can give  a much simpler, direct  algorithm permitting to find all weighted asteroidal triples in time $O(n^3)$ for a $n \times n$ symmetric matrix $A$.
 
 Finally we mention a possible application of our characterization for identifying large Robinsonian submatrices. In particular we obtain an explicit characterization of the maximal subsets $I$ for which the principal submatrix $A[I]$ is Robinsonian in terms of forbidden `weighted asteroidal cycles'.

\subsection{Recovering the structural characterization of unit interval graphs}\label{secuig}

In this section we indicate how to recover from our main result (Theorem \ref{thm:main})  the known structural characterization of unit interval graphs in Theorem \ref{thm:uig}. 
%We recall some definitions.

Let $G=(V,E)$ be a graph.
Given $x,y,z\in V$ and a path $P$ from $x$ to $y$ in $G$,  recall that $P$ misses $z$ if $P$ is disjoint from the neighborhood of $z$.
% i.e., if $P=(x=x_0,x_1,\cdots, x_k, y=x_{k+1})$ then all pairs $\{x_i,x_{i+1}\}$ ($0\le i\le k$) are edges of $G$ and $z$ is not adjacent to any node of $P$. 
In other words, if  $A_G$ denotes  the adjacency matrix of $G$,  then $P=(x=x_0,x_1,\cdots, x_k, y=x_{k+1})$ misses $z$ if $(A_G)_{x_i,x_{i+1}}>\max\{(A_G)_{x_iz},(A_G)_{x_{i+1}z}\}$ for all $0\le i\le k$.
Hence if $P$ misses $z$,  then it also avoids $z$ in $A_G$, but the converse is not true in general.

An asteroidal triple in $G$ is  a set of nodes $\{x,y,z\}$ 
containing  no edge and 
such that  there exists a path in $G$ between  any two nodes in $\{x,y,z\}$  that misses the third one. 
%We will call such a triple a {\em graphic asteroidal triple} in order to distinguish with the notion of {\em weighted asteroidal triple}.
Hence, if $\{x,y,z\}$ is an asteroidal triple in $G$,  then it is also a weighted asteroidal triple in the adjacency matrix $A_G$ of $G$, but the converse is not true in general.

\ignore{

\begin{theorem}\label{thm:uig}\cite{}
A graph $G$ is a unit interval graph if and only if it satisfies the following conditions:
\begin{description}
\item[(i)] $G$ is chordal, i.e., $G$ contains no induced cycle of length at least 4;
\item[(ii)] $G$ contains no induced claw $K_{1,3}$;
\item[(iii)]
$G$ contains no (graphic) asteroidal triple.
\end{description}
\end{theorem}
}

As was recalled earlier, $G$ is a unit interval graph if and only if its adjacency matrix $A_G$ is Robinsonian.
In view of Theorem \ref{thm:main}, $A_G$ is Robinsonian if and only if it does not contain a weighted asteroidal triple.
Combining those two facts with Theorem~\ref{thm:uig}, we have the following.

\begin{theorem}\label{thm:at}
Let $G$ be a graph and $A_G$ its adjacency matrix. 
The conditions (i), (ii), and (iii) in Theorem~\ref{thm:uig} hold if and only if $A_G$ has no weighted asteroidal triple. 
\end{theorem}

We now  give a direct  proof of Theorem~\ref{thm:at}, which in turn implies an alternative proof of Theorem~\ref{thm:uig}.
%
%In order to derive Thorem \ref{thm:uig} it suffices to show that one of the conditions (i),(ii),(iii) in Theorem \ref{thm:uig} is violated if and only if  a weighted asteroidal triple can be found in $A_G$, which is what we do in Lemmas \ref{lem:uig1} and \ref{lem:uig2} below.

\begin{lemma}\label{lem:uig1}
Let $G$ be a graph and $A_G$ its adjacency matrix. Assume that one of the conditions (i), (ii), or (iii) in Theorem \ref{thm:uig} is violated. Then one can find a weighted asteroidal triple in $A_G$.
\end{lemma}

\begin{proof}
Assume first that we have an induced  cycle $C=(x_1,\cdots,x_k)$ of length $k\ge 4$ in $G$. Then, $(x_1,x_2)$ avoids $x_k$ in $A_G$ and thus 
$x_1\overset{x_k}{\sim}x_2$, and $(x_1,x_k)$ avoids $x_2$ in $A_G$ and thus $x_1\overset{x_2}{\sim} x_k$. In addition,
the path $(x_2,x_3,\cdots, x_{k-1},x_k)$ avoids $x_1$ in $A_G$, which gives $x_2\overset{x_1}{\sim}x_k$. Hence $\{x_1,x_2,x_k\}$ is a weighted asteroidal triple in $A_G$.

Assume now that we have a claw $K_{1,3}$ in $G$, say $u$ is adjacent to $x,y,z$ and $\{x,y,z\}$ is independent in $G$.
Then, $x\overset{z}{\sim} u\overset {z}{\sim } y$, $x\overset{y}{\sim} u\overset {y}{\sim } z$,
$y\overset{x}{\sim} u\overset {x}{\sim } z$ in $A_G$, and thus $\{x,y,z\}$ is a weighted  asteroidal triple in $A_G$.

Finally, if $\{x,y,z\}$ is an asteroidal triple in $G$ then  clearly it is also a weighted asteroidal triple in $A_G$.
\end{proof}

To prove the converse we will use the following result.

\begin{lemma}\label{lem:uig0}
Let $G$ be a graph with adjacency matrix $A_G$ and consider distinct elements $x,y,z\in V$.
Assume $P$ is a path from $x$ to $y$ avoiding $z$ in $A_G$ which has the smallest possible number of nodes. Then one of the following holds:
\begin{description}
\item[(i)]  $P$ is a path in $G$ that misses $z$ (i.e., $z$ is not adjacent to any node of $P$);
\item[(ii)] we find a claw or an induced cycle of length at least 4 in $G$;
\item[(iii)] $P$ is an induced path in $G$,  $z$ is adjacent to exactly one node $u$ of $P$ and  $u\in \{x,y\}$. %of $x$ and $y$.
\end{description}
\end{lemma}

\begin{proof}
Let $P=(x=x_0,x_1,\cdots, x_k,x_{k+1}=y)$ be  a path satisfying the assumptions of the lemma.
As $P$ avoids $z$ in $A_G$, it follows that  $z$ cannot be adjacent to two consecutive nodes in $P$. If $z$ is not adjacent to any node of $P$ then  we are in case (i). Assume first  that $z$ is adajcent to at least two nodes of $P$. Say, $z$ is adjacent to $x_i$ and $x_j$ with $0\le i\le j-2 \le k-1$, and $z$ is not adjacent to any $x_h$ with $i<h<j$. Consider  the subpath $(x_i,\cdots, x_j)$ of $P$.
If  this subpath is not induced in $G$ then we could find a shorter path than $P$ going from $x$ to $y$ and avoiding $z$ in $A_G$, contradicting our minimality assumption on $P$.
Hence this subpath is induced, so we find an induced cycle of length at least 4 and we are in case (ii).

We can now assume that $z$ is adjacent to exactly one node $x_i$ of $P$. Then the path $P$ is induced in $G$ (for if not one would contradict the minimality of $P$). 
If $x_i$ is not the first or last node of $P$ then we find a claw and thus we are again in case (ii). 
Hence we can conclude that $z$ is adjacent to exactly one of $x$ and $y$. Hence  we are in case (iii).
\end{proof}

\begin{lemma}\label{lem:uig2}
Let $G$ be a graph and $A_G$ its adjacency matrix.
If there exists a weighted asteroidal triple in $A_G$ then one of the conditions (i), (ii), (iii) in Theorem \ref{thm:uig} is violated.
\end{lemma}

\begin{proof}
Assume that $\{x,y,z\}$ is  a weighted asteroidal triple  in $A_G$. Select paths $P_{xy}$, $P_{xz}$ and $P_{yz}$ that avoid, respectively, $z$, $y$, $x$ in $A_G$ and have the smallest possible lengths. We apply Lemma \ref{lem:uig0} to each of the three paths $P_{xy}$, $P_{xz}$ and $P_{yz}$. 
If for some of these three paths we are in case (ii) of Lemma \ref{lem:uig0}, then we find a claw or an induced cycle and we are done. 
Hence, for each of  the three paths  we are  in case (i) or (iii) of Lemma \ref{lem:uig0}.
If for all the three paths we are in case (i),  then we can conclude that $\{x,y,z\}$ is an asteroidal triple of $G$ and we are done.

Therefore, we may now assume that (say) for the path $P_{xy}$, we are in case (iii). Then $P_{xy}$ is an induced path in $G$ and (say) $z$ is adjacent to $x$.
In turn, this implies that we are in case (iii) also for the path $P_{yz}$ and thus $P_{yz}$ is induced in $G$. Together with the edge $\{x,z\}$ the two paths $P_{xy}$ and $P_{yz}$ form a cycle $C$ with at least 4 nodes. If $C$ is induced in $G$ then we are done. So let us now assume that $C$ has a chord.
As both paths $P_{xy}$ and $P_{yz}$ are induced there must exist an edge of the form $\{u,v\}$ where $u$ belongs to $P_{xy}$ and $v$ belongs to $P_{yz}$.
First we choose $u$ to be the `first' node on $P_{xy}$   which is adjacent to some node $v$ of $P_{yz}$, `first' when travelling from $x$ to $y$ on $P_{xy}$. 
After that,  we choose for  $v$ the `last'  node of $P_{yz}$ which is adjacent to $u$, `last' when travelling from $y$ to $z$ on $P_{yz}$. 
Note that $u\neq x$, $u\neq y$, $v\neq z$, and $v\neq y$ since we are in case (iii).
Thus, from the choice of $u$ and $v$, it  follows that the cycle obtained by travelling along $P_{xy}$ from $x$ to $u$, then traversing edge $\{u,v\}$, then travelling along $P_{yz}$ from $v$ to $z$, and finally traversing edge $\{z,x\}$, is an induced cycle in $G$ of length at least 4. This concludes the proof.
\end{proof}

Combining Lemmas \ref{lem:uig1} and \ref{lem:uig2}  completes the proof of Theorem~\ref{thm:at}.

\subsection{An algorithm for enumerating the weighted asteroidal triples}

As an application of our main  theorem  we obtain an alternative algorithm to decide whether a given matrix $A$ is Robinsionan, namely
by checking  the existence of a weighted asteroidal triple for $A$. We  indicate a simple algorithm for doing this.
 
\medskip
A first observation is that, given distinct elements $x,y,v\in V$, one can check whether $x\overset{v}{\sim} {y}$, i.e., whether there exists a path from $x$ to $y$ avoiding 
$v$ in $A$, in time $O(n^2)$. For this consider the graph $H_v$ with vertex set $V\setminus \{v\}$, where two nodes $u,w\in V\setminus \{v\}$ are adjacent if 
$A_{vw}>\min\{A_{uv},A_{vw}\}$. Then $x\overset{v}{\sim} {y}$ precisely when $x$ and $y$ lie in the same connected component of $H_v$.
Building $H_v$ and checking the existence of a path from $x$ to $z$ in $H_v$ can be done in time $O(n^2)$.

A first elementary algorithm to decide existence of a weighted asteroidal triple in $A$ would be to test all possible triples, which can be done in time $O(n^5)$.
The following simple algorithm permits to check existence of a weighted asteroidal triple more efficiently, in time $O(n^3)$.
It computes  a function $f$ defined on the set $V\choose 3$ of all triples of elements of $V$, whose value records whether a triple is a weighted  asteroidal triple.

%Given three elements $x, y, z\in V$, one can easily decide whether $x\overset{z}{\sim} y$ or not in $O(n^2)$ time, 
%and hence a native algorithm can decide whether $A$ has a weighted asteroidal triple in $O(n^5)$ time. 
%A slightly elaborate algorithm given in Algorithm~\ref{alg1} implements the idea in $O(n^3)$ time.
%As a corollary of Theorem~\ref{thm:main}, this simple algorithm decide whether $A$ is Robinsonian or not in $O(n^3)$ time.

\begin{algorithm}
  \caption{} \label{alg1}
  \SetKwInput {KwIn}{input}
  \SetKwInput {KwOut}{output}
 \KwIn{a symmetric matrix $A$ (indexed by $V$)}
 \KwOut{A weighted asteroidal triple $\{x,y,z\}$ or ``$A$ has no weighted asteroidal triple".}
  \vspace{2ex}  
  Initialize $f:{V\choose 3}\rightarrow \mathbb{Z}$ by $f(\{x, y, z\})=0$ for every $\{x,y,z\}\in {V\choose 3}$.\\
  \For{each $v \in V$}{
    Compute the graph  $H_v$ with  vertex set $V\setminus \{v\}$ and with edges the pairs $\{u,w\}$ such that  $A_{uw}>\min\{A_{uv},A_{vw}\}$.\\
    %$E(G_v)=\{xy: x\overset{v}{\sim} y\}$.
	$f(\{x, y, v\})\leftarrow f(\{x,y,v\})+1$ for every pair $\{x, y\}$ of elements lying in the same  connected component of $H_v$.
  }
   \eIf{there exists $ \{x,y,z\}\in{V\choose 3}$ with $f(\{x,y,z\})=3$}{
   	 	\Return{$\{x,y,z\}$}
   	}{ 
    		\Return{``$A$ has no weighted asteroidal triple"}
    }
\end{algorithm}

In fact the final function $f$ returned by the above algorithm permits to return all the weighted asteroidal triples, which are  precisely the triples $\{x,y,z\}$ with $f(\{x,y,z\})=3$.

\subsection{Maximal Robinsonian submatrices}

When  a given symmetric matrix $A$ indexed by $V$ is not Robinsonian,  one might be interested in the maximal subsets indexing a Robinsonian submatrix or, equivalently, in the minimal subsets whose deletion leaves a Robinsonian submatrix.  Note that finding a Robinsonian submatrix of largest possible size  is in fact a hard problem, already for binary matrices.  Indeed it is known that  finding in a given graph  a smallest cardinality set of nodes whose deletion leaves a unit interval graph is an NP-complete problem (see \cite{Lewis80,Hajiaghayi02}).

Let $\MI_A$ denote the collection of all maximal subsets $I\subseteq V$ for which $A[I]$ is a Robinsonian matrix 
%$$\MI_A=\{I\subseteq V: A[I] \text{ is Robinsonian}\},$$
%let $\MB_A$ consist of the maximal sets in $\MI_A$  
and let $\MF_A$ consist of the minimal subsets $F\subseteq V$ for which $A[V\setminus F]$ is Robinsonian, i.e., 
$\MF_A=\{V\setminus I: I\in \MI_A\}.$ 
%Then $\MI_A$  and $\MF_A$ are antichains (i.e., no two elements are comparable for inclusion). One can also characterize  $\MI_A$ in terms of its collection $\MC_A$ of {\em circuits}, which are the minimal subsets $C\subseteq V$ that do not belong to $\MI_A$. Equivalently, $C\in \MC_A$ if and only if  $C$ is a minimal transversal  of the collection  $\MF_A$.
Let also $\MC_A$ denote the collection of minimal transversals of $\MF_A$ (i.e., the minimal sets intersecting all sets in $\MF_A$).
In other words, $\MI_A$ coincides with the collection of maximal independent sets of the hypergraph $\MH_A=(V,\MC_A)$, whose dual (or transversal) hypergraph is  $\MH_A^d=(V,\MF_A)$ (see, e.g., \cite{FK}).

In order to describe the minimal transversals of $\MF_A$   we introduce the  following definition. A set $C\subseteq V$ is called a {\em weighted asteroidal cycle} of $A$  if there exists a weighted asteroidal triple $\{x,y,z\}$ of $A$ and paths $P_{xy}$, $P_{xz}$, $P_{yz}$ such that $C=V(P_{xy})\cup V(P_{xz})\cup V(P_{yz})$, where $P_{xy}$ (resp., $P_{xz}$, $P_{yz}$) is a path from $x$ to $y$ avoiding $z$ (resp., from $x$ to $z$ avoiding $y$, from $y$ to $z$ avoiding $x$). Then, as a direct application of Theorem \ref{thm:main}, for  sets $I,C\subseteq V$, we have:
$$A[I] \text{ is Robinsonian } \Longleftrightarrow \ I \text{ does not contain a weighted asteroidal cycle},$$
%and,  for a set $C\subseteq V$, 
$$C\in \MC_A \ \Longleftrightarrow \ C \text { is a minimal weighted asteroidal cycle of } A. $$

Furthermore, a set  $C$ is a minimal weighted asteroidal cycle of $A$ if and only if $A[C]$ is not Robinsonian, but $A[C\setminus \{x\}]$ is Robinsonian for any $x\in C$,
 and thus one can check in polynomial time membership in the collection $\MC_A$. Analogously, one can  check membership in the collection $\MI_A$ in polynomial time,
  since $I\in \MI_A$ if and only if $A[I]$ is Robinsonian, but $A[I\cup \{x\}]$ is not Robinsonian for any $x\in V\setminus I$. 

\medskip
As mentioned above,  it is of interest to  generate the elements of $\MI_A$ (which correspond to the   maximal Robinsonian submatrices of $A$)
as well as the sets in $\MC_A$ (which correspond to the minimal obstructions to the Robinsonian property).
For this one can apply the algorithmic approach developed in \cite{GK}, which gives a quasi-polynomial time incremental   algorithm for the joint generation of the collections $(\MF_A, \MI_A)$. 
Namely, given $\MX\subseteq \MF_A$ and $\MY\subseteq \MI_A$, the algorithm of \cite{GK} permits to decide whether $(\MX,\MY)=(\MF_A,\MI_A)$ and if not to find a new set in $(\MF_A\setminus \MX) \cup (\MI_A\setminus \MY)$, in time  $O(n^4 +m^{O(\log m)})$, where 
$m=|\MF_A|+|\MI_A|$.
Then, starting $(\MX,\MY)=(\emptyset, \emptyset)$, this incremental algorithm can find $(\MF_A,\MI_A)$ in $|\MF_A|+|\MI_A|$ iterations.
%The  algorithm is presented in \cite{BI,GK} for the following problem: 

\medskip
On a more practical point of view, when $A$ is not Robinsonian,  one may  try to remove some of its rows/columns and/or modify some of its entries in order to eliminate its weighted asteroidal cycles. 
Investigating whether this may lead to useful heuristics to get good Robinsonian approximations will be the  subject of future research.

%\medskip
%{\bf To do:} Add the reference that finding the max size of a Robinsonian submatrix is NP-complete already in the graph case.

\subsection*{Acknowledgements}
%M.L. and M.S. 
The first two authors thank Utz-Uwe Haus for useful discussions and for bringing to our attention the works \cite{FK,GK}.
The third author was supported by JSPS Postdoctoral Fellowships for Research Abroad and JSPS Grant-in-Aid for Young Scientists (B) (No. 15K15942).

%\bibliography{Robinsonian}
%\bibliographystyle{plain}

\end{document}